\documentclass{lipics-v2021}

\usepackage{xspace}

\usepackage{thmtools}
\usepackage{amsmath, amsfonts, amssymb}
\usepackage{amsthm}
\usepackage{tikz}
\usepackage{mathdots}
\usepackage{cancel}
\usepackage{color}
\usepackage{url}
\usepackage{hyperref}
\hypersetup{
    colorlinks=true,
    linkcolor=black,
    citecolor=black,
    filecolor=black,
    urlcolor=black,
}
\usepackage[capitalize]{cleveref}
\crefname{step}{step}{steps}
\Crefname{step}{Step}{Steps}

\usepackage{siunitx}
\usepackage{array}
\usepackage{multirow}
\usepackage{gensymb}
\usepackage{tabularx}
\usepackage{extarrows}
\usepackage{booktabs}
\usepackage{subcaption}
\usepackage{todonotes}
\usetikzlibrary{fadings}
\usetikzlibrary{patterns}
\usetikzlibrary{shadows.blur}
\usetikzlibrary{shapes}
\usepackage[skins]{tcolorbox}
\usepackage{enumerate}
\usepackage{cleveref}
\usepackage{algorithm}
\usepackage{algpseudocode}
\usepackage{xcolor} 
\usepackage{tcolorbox}

\usepackage{thm-restate}

\newcommand{\N}{\mathbb{N}}

\newcommand{\id}{\textsc{id}}

\newcommand{\LOCAL}{\textsf{LOCAL}}

\newtcolorbox{algorithmbox}[2][]{%
	enhanced,colback=white,colframe=black,coltitle=black,
	sharp corners,boxrule=0.4pt,
	fonttitle=\itshape,
	attach boxed title to top left={yshift=-0.3\baselineskip-0.4pt,xshift=2mm},
	boxed title style={tile,size=minimal,left=0.5mm,right=0.5mm,
		colback=white,before upper=\strut},
	title=#2,#1
}

\definecolor{lightred}{rgb}{1,0.6,0.6}

\newtcolorbox{box1}[1][]{colback=gray!20, colframe=black, coltext=black, boxrule=1pt, sharp corners, #1}

\newcommand{\problem}{Increasing Degree problem\xspace}

\newcommand{\problemp}{Partition problem\xspace}

\usepackage{nicefrac}
\usepackage{todonotes}

\newcommand{\fil}[1]{\todo[color=blue!20]{\tiny F: #1}}

\newcommand{\ldeg}{\sharp\mbox{leaf}}
\usepackage{xcolor}
\usepackage[draft]{fixme}
\usepackage{algorithm}
\usepackage{algpseudocode}

\fxusetheme{color}
\fxsetface{inline}{\small}

\renewcommand{\AA}{\mathcal{A}}

\newcommand{\LL}{\mathcal{L}}

\newcommand{\TT}{\mathcal{T}}

\newcommand{\logstar}{\log^{\star}}

\begin{document}

\title{A Simple Construction of Locally Checkable Problems Filling the LOCAL Complexity Gaps in Graphs with Arbitrary Large Degrees}

\titlerunning{A Simple Construction of Locally Checkable Problems Filling the Complexity Gaps}

\author{Filippo Casagrande}
{Gran Sasso Science Institute, Italy}
{}
{}
{}

\author{Pierre Fraigniaud}
{Institut de Recherche en Informatique Fondamentale (IRIF)\\
CNRS and Université Paris Cité, France}
{}
{}
{}

\author{Benjamin Jauregui}
{- Institut de Recherche en Informatique Fondamentale,
CNRS and Université Paris Cité, France\\
- Departamento de Ingeniería Matemática, Universidad de Chile, Chile}
{}
{}
{}

\author{Mikaël Rabie}
{Institut de Recherche en Informatique Fondamentale  (IRIF)\\ CNRS and Université Paris Cité, France}
{}
{}
{}

\authorrunning{F. Casagrande, P. Fraigniaud, B. Jauregui, M. Rabie }



\funding{The second and fourth authors are supported by ANR Project ENEDISC (ANR-24-CE48-7768-01).}

\hideLIPIcs

\nolinenumbers

\acknowledgements{The authors are thankful to Alkida Balliu, Francesco d'Amore and Dennis Olivetti for their precious comments and advices, and to Ami Paz for preliminary discussions on the topic of this paper.}

\maketitle

\begin{abstract}
    
    We show that the complexity gaps in the round complexities of \emph{locally checkable labeling} (LCL) problems are not due to the fact that solutions to LCL problems must be locally checkable, but solely to the fact that LCL problems are defined only for graphs of maximum degree upper bounded by some arbitrary yet constant value~$\Delta$. Specifically, we show that there are infinitely many locally checkable problems (i.e., problems whose solutions can be checked locally) whose round complexities belongs to the two intervals $[\omega(1),o(\log\log^\star n)]$ and $[\omega(\log^\star n),o(\log n)]$ whenever these problems are considered in networks with unbounded maximum degrees. This extends the previous results by Schmid (arXiv, 2026), which hold for the polynomial regime only, and by Bousquet, Feuilloley, and Pierron (OPODIS, 2025), which hold for trees only. 
        
    
    All our upper bounds are obtained using deterministic algorithms that can be run under the \emph{port-numbering} model, which is a weak variant of \LOCAL, without any a priori information on the number of nodes in the network. Instead, our lower bounds apply to randomized \LOCAL, and quantum \LOCAL, even if nodes have identifiers in $[1,n]$, and even if they know the exact number of nodes in the network. They even hold under randomized online \LOCAL, a strong variant of the \LOCAL\ model. 
    Finally, our lower bounds  hold even for trees. 
    

    Our results are obtained using two main ingredients. The first one is the analysis of a new locally checkable problem called \emph{Increasing Degree}, parameterized by a function $f:\mathbb{N}\to\mathbb{N}$. Different round complexities can be obtained by tuning the function $f$ accordingly. Our second tool is a general \emph{Translation Theorem} that enables to transfer results from a given range of complexities to results for a  range of lower complexities. 
    
\end{abstract}

\thispagestyle{empty}
\newpage 
\setcounter{page}{1}

\section{Introduction}

\subsection{Locality and LCL Problems}

The study of \emph{spatial locality} in distributed computing has been mostly carried out under the so-called \LOCAL\ model~\cite{linial1992locality,peleg2000distributed}. Recall that this model assumes a set of $n\geq 1$ processing elements connected by a network modeled as a simple connected $n$-node graph $G=(V,E)$. Each node $v\in V$ models a processing element provided with an identifier $\id(v)\in \mathbb{N}$, which is unique in the network. It is often assumed that the identifiers belong to an interval $[0,N]$ where $N=\mbox{poly}(n)$ is a polynomial upper bound on the number of nodes $n$ in the network. The bound $N$ may be given to the nodes as an initial knowledge. Other than that, the nodes do not have information about the graph they belong to. Nodes exchange messages along the edges of the graph, which model communication links. Computation proceeds as a sequence of synchronous rounds, all nodes starting at the same round. In each round, every node performs some local computation, then sends a message to each of its neighbors in the network, and receives the messages sent by its neighbors. 

Under the \LOCAL\ model, particular attention has been paid to what is known as \emph{locally checkable labeling} (LCL) problems~\cite{naor1993can}. The appeal of this model lies in two features. First, it provides \emph{finite descriptions} for  a vast family of problems. Second, the global correctness of a solution to such a problem is the conjunction of \emph{local correctness} conditions at all nodes. That is, solutions to LCL problems are easy to check, and the question is: Are they easy to compute? 

More formally, an LCL is described by a triple $(\Delta,L,\mathcal{B})$ where $\Delta$ is a positive integer, $L$ is a finite set whose elements are called \emph{labels}, and $\mathcal{B}$ is a finite collection of \emph{$L$-labeled balls} of radius $r$ in graphs of maximum degree~$\Delta$, for some integer $r\geq 0$. Recall that the ball $B_G(v,r)$ centered at $v$ and of radius $r$ in a graph $G=(V,E)$ is the subgraph induced by all the nodes at distance at most $r$ from $v\in V$. A ball is $L$-labeled if each node $w$ of the ball is labeled by a label $\ell(w)\in L$. Given an LCL $(\Delta,L,\mathcal{B})$, a distributed algorithm $\mathsf{ALG}$ solves the corresponding problem $\Pi(\Delta,L,\mathcal{B})$ if, given any graph $G=(V,E)$ of maximum degree~$\Delta$, $\mathsf{ALG}$ outputs a label in $L$ at each node such that, for every $v\in V$, the labeled ball $B_G(v,r)\in\mathcal{B}$. 

Note that the specification of LCL problems is very general and enables to capture many classical graph problems such as proper vertex- or edge-coloring, maximal independent set, maximal matching, etc. Note also that one may assume  input labels given to the nodes, in which case the correctness of an output labeling may depend on the given input labeling. A crucial point is that an algorithm $\mathsf{ALG}$ aiming at solving an LCL problem $\Pi$ is designed for running specifically in graphs of maximum degree~$\Delta$, where $\Delta$ is a fixed constant, as $\mathsf{ALG}$ is aiming at solving~$\Pi$, which is specified only for graphs with a certain maximum degree~$\Delta$. In the \LOCAL\ model, the \emph{round complexity} of  $\mathsf{ALG}$ is thus the maximum, taken over all $n$-node graphs $G$ of maximum degree~$\Delta$, of the number of rounds executed by $\mathsf{ALG}$ until all nodes of $G$ output, and the complexity of an LCL problem $\Pi$ is the minimum, taken over all algorithms $\mathsf{ALG}$ solving $\Pi$, of the round complexity of $\mathsf{ALG}$. 

Very many interesting facts can be stated for LCL problems, including the fact that it is undecidable whether they can be solved in a constant number of rounds~\cite{naor1993can}, and the fact that the randomized complexity of an LCL problem on instances of size $n$ is at least the deterministic complexity of the problem on instances of size $\Theta(\sqrt{\log n})$~\cite{chang2019exponential}, which is essentially tight up to polylogarithmic factors. (This contrasts with the open question  $\mathsf{P}$ vs. $\mathsf{BPP}$ in sequential computing.) Yet, one of the most striking fact about solving LCL problems in the \LOCAL\ model is the presence of wide \emph{complexity gaps}, contrasting with the Time Hierarchy Theorem for Turing machines. 

\subsection{Complexity Gaps}

It is known~\cite{chang2019time} that no LCL problem has deterministic round complexity in the interval $[\omega(\log^\star n),o(\log n)]$ where $\log^\star n$ denotes the number of times the logarithm function should be applied to $n$ until the result is at most~1. 
In other words, if an LCL problem has complexity $o(\log n)$, then it has complexity $O(\log^\star n)$. 
The fact that no LCL problem can have complexity strictly between $\Theta(\log^\star n)$ and $\Theta(\log n)$ came as a stunning surprise, but there is a clear explanation for this. 
The idea is to \emph{lie about $n$}: many algorithms have a round complexity that depends on $n$ only because identifiers range over $\mathrm{poly}(n)$. 
If we could instead give a \emph{virtual identifiers} to every node from a range depending only on $\Delta$, and run the algorithm on these, its complexity would no longer depend on $n$. Of course, in this way, many nodes will share the same virtual identifier.
Coloring is used to avoid collisions while running the algorithm with this small amount of new identifiers: the algorithm do not have to find multiple nodes with the same identifier. 
There is a deterministic algorithm  with round complexity $O(\log^\star n)$ that solves proper $O(\Delta^2)$-coloring in graphs of maximum degree~$\Delta$ (see~\cite{linial1992locality}). 
Applying it to the $k$th power graph $G^k$ of $G$, where $k$ depends on the LCL problem at hand, this gives a proper coloring of $G$ where two nodes that share the same color are far away (with distance at least $k$). 
This enables us to run the algorithm using the colors as virtual identifiers instead of the original ones, and still producing a correct output, because the global correctness is the conjunction of local correctness conditions.

Interestingly, another complexity gap pops up for lower complexities~\cite{chang2019time,naor1993can}, between $\Theta(\log\log^\star n)$ and $O(1)$. That is, if an LCL problem has deterministic complexity $o(\log\log^\star n)$, then it has constant deterministic complexity. Again, the fact that the considered graphs have maximum degree~$\Delta$ plays an important role here. Roughly, Ramsey theory enables us to show that if there is an algorithm with round-complexity $o(\log\log^\star n)$ solving an LCL problem $\Pi$, then there is an \emph{order invariant} algorithm (i.e., an algorithm using only the relative order of the identifiers, not their actual values) solving $\Pi$ in $o(\log\log^\star n)$ rounds. 
The result essentially follows from the fact that any $o(\log \log^* n)$-round algorithm can be transformed into one whose behavior depends only on the relative order of node identifiers, rather than their numerical values. This makes it possible to iteratively compress the identifier space down to $O(1)$ while preserving the algorithm's behavior.

Furthermore, similar complexity gaps occur for \emph{randomized} algorithms, which are Monte Carlo algorithms succeeding with high probability, i.e., probablility at least $1-O(1/n)$.  For instance, no LCL problem has randomized round complexity in the range $[\omega(1),o(\log\log^\star n)]$, or in the range $[\omega(\log^\star n),o(\log\log n)]$~\cite{suomela2020landscape}. 
The question addressed in this paper is the following: 
\begin{center}
    \emph{Do these complexity gaps persist if we relax the restriction that \\ we must work only with graphs of bounded degree?} 
\end{center}
It is worth mentioning three recent contributions that provide partial (positive and negative) answers to this question. 
\begin{itemize}
    \item First, it was proved by Schmid~\cite{schmid2026lcls} that all round complexities of the form $n^\epsilon$ for $\epsilon\in(0,1]$ exist whenever the degrees are unbounded. However, it is not known whether the complexity gaps also disappear below the polynomial regime, in particular for the two gaps $[\omega(1),o(\log\log^\star n)]$ and $[\omega(\log^\star n),o(\log n)]$. 
    
    \item Second, it was shown by Bousquet, Feuilloley and Pierron~\cite{bousquet2024local} that the answer is also negative if one restricts our analysis to trees. LCL problems in \emph{trees} (of bounded maximum degree~$\Delta$) can have only the following deterministic complexity: $O(1)$, $\Theta(\log^\star n)$, $\Theta(\log n)$, and $\Theta(n^{1/k})$ for all integers $k\geq 1$~\cite{grunau2022landscape}. Bousquet et al.~\cite{bousquet2024local} have proved that removing the constraints on the maximum degree enables to define locally checkable problems with \emph{arbitrary} round-complexities in trees. 

    \item In contrast to the above two results, it was very recently shown by Piombi~\cite{Piombi2026} that the complexity gaps persist even for \emph{rooted} trees of unbounded degrees, whenever the locally checkable problems are defined by local Presburger monadic second-order formulas. 
\end{itemize}
The objective of this paper is to study whether the result in~\cite{bousquet2024local} extends from trees to arbitrary graphs, beyond the case of the polynomial regime studied in~\cite{schmid2026lcls}, or whether the results in~\cite{Piombi2026} extends beyond local Presburger monadic second-order formulas.

\subsection{Our Results}

We focus on \emph{locally checkable problems}, which are defined roughly as LCL problems, but where the requirement of working with graphs of maximum degree $\Delta$ is dropped. We however keep the requirement that the global correctness of a solution is defined as the conjunction of a local correctness condition at all nodes. That is, we consider problems defined by a pair $(L,\varphi)$ where $L$ is a finite set of labels, and $\varphi$ is a Boolean predicate applying to $L$-labeled balls of radius $r$ in arbitrary graphs, for some integer~$r\geq 0$. Solving such a problem asks for labeling every node with a label in $L$ such that the $L$-labeled ball of radius $r$ centered at every node of the input graph satisfies~$\varphi$. We do not enter into the details of whether the predicate $\varphi$ is expressible in a certain logic (e.g., first-order logic) as we shall explicitly work with specific problems for which the predicate $\varphi$ is explicit, which suffices for the purpose of this paper.

In short, the take away message of this paper is that the aforementioned complexity intervals  $[\omega(1),o(\log\log^\star n)]$ and $[\omega(\log^\star n),o(\log n)]$  are far from being empty whenever the bounded degree condition is dropped. More specifically, we establish the following set of results. 

\begin{itemize}
    \item For the interval $[\omega(\log^\star n),o(\log n)]$, we describe problems $\Pi_f$ parameterized by functions ${f:\mathbb{N}\to\mathbb{N}}$, and we show how to choose $f$ so that to produce various locally checkable problems with many different complexities in $[\omega(\log^\star n),o(\log n)]$. In particular, we construct problems with complexities $\Theta(\log^{(k)}n)$ for all $k\geq 1$, where $\log^{(k)}$ denotes the $k$th iterated logarithm\footnote{For every integer $k\geq 1$, the  $k$th iterate of a function $f$ is the function $f^{(k)}$ defined as $f^{(1)}=f$, and $f^{(k)}=f\circ f^{(k-1)}$ for $k\geq 2$ (e.g., $\log^{(2)}n=\log\log n$).}, and problems with complexities $\Theta(\log^{1-\epsilon}n)$ for all $\epsilon\in(0,1)$. Interestingly, we also provide a locally checkable problem with complexity $\Theta(\log^\star n)$ for which there exists an optimal algorithm that does not use identifiers. This is in contrast to LCL problems with complexity $\Theta(\log^\star n)$ whose solutions typically heavily use the node identifiers. To our knowledge, this is the first problem with complexity $\Theta(\log^\star n)$ in the \emph{port-numbering} variant of the \LOCAL\ model~\cite{angluin1980local}, in which nodes do not have identifiers, but the incident edges of every degree-$d$ node are labeled from 1 to $d$ in arbitrary order.

    \item For dealing with the interval $[\omega(1),o(\log\log^\star n)]$, we construct a ``translation operator'' $\Phi$ that enables us to transform any problem $\Pi_f$ with complexity in $[\Omega(\log^\star n),O(\log n)]$ into a  problem $\Pi_{\Phi(f)}$ with complexity in ${[\Omega(\log^\star\log^\star n), O(\log \log^\star n)]}$. The interest of this operator is that it can be iterated to obtain a problem $\Pi_{\Phi(\Phi(f))}$ with complexity in ${[\Omega(\log^\star\log^\star\log^\star n),O(\log\log^\star\log^\star n)]}$, and so on. Since the function $\log^\star$ is successively applied in this construction, all the resulting complexities remain in $\omega(\log^\sharp n)$, where $\log^\sharp n$ denotes the number of times the function $\log^\star$ must be applied to $n$ so that to obtain a value at most 1. 
    Nevertheless, we also design a problem with round complexity $\Theta(\log^\sharp n)$. 
\end{itemize}

\noindent Figures~\ref{fig:comp_functions} and~\ref{fig:comp_functions2} provides a compact representation of our results. 
It is worth pointing out the following facts resulting from our analysis. 

\begin{description}
    \item[Robust upper and lower bounds.] Our lower bounds hold even for trees, while our upper bounds hold for arbitrary graphs. Moreover, our lower bounds hold even if the nodes are given identifiers in $[0,n-1]$, and even if $n$ is given to the nodes as initial knowledge. Instead, our upper bounds do not assume any given bound on the number of nodes, and they hold even in the absence of identifiers. That is, our upper bounds hold in the aforementioned port-numbering variant of the \LOCAL\ model~\cite{angluin1980local}. In particular, we provide a problem with complexity $\Theta(\log^*n)$ that does not use identifiers nor the knowledge of the size of the network to be solved, which was, to our knowledge, not known before.

    \item[Randomized complexity, and beyond.]
    Our upper bounds are obtained using deterministic algorithms, but our matching lower bounds also hold for randomized algorithms. In fact, they also hold for \emph{non-signaling finitely-dependent} probability distributions, which implies that they hold for \emph{quantum} algorithms too, i.e., for algorithms where processing nodes can prepare quantum states, perform quantum operations on these states, and exchange qubits with their neighbors via the communication links of the network. Finally, our lower bounds also hold in the randomized online variant of the \LOCAL\ model, one of the strongest variants of \LOCAL. 
\end{description}

To sum up the above remarks, our upper bounds use deterministic algorithms, and they apply to all graphs under the port-numbering variant of the \LOCAL\ model in absence of any information about the input graph given to the nodes. In contrast, our matching lower bounds hold for the quantum variant of the \LOCAL\ model, as well as for the randomized online variant \LOCAL, even in trees, and even if the nodes are given identifiers in $[0,n-1]$ in $n$-node trees. 

\begin{figure}[tb]
    \centering
    \scalebox{0.8}{\input{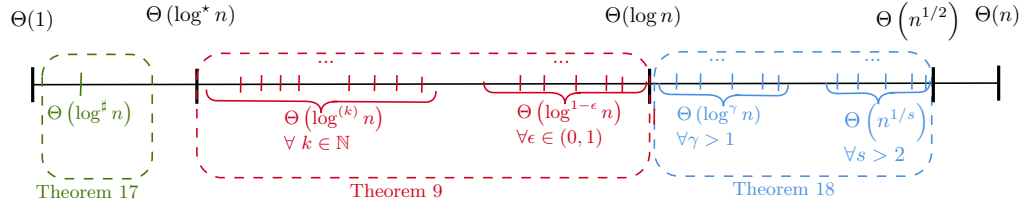}}
    \caption{New complexities obtained directly by analyzing the complexity of specifically designed problems, tuned thanks to appropriate choices of the parameter function~$f$. }
    \label{fig:comp_functions}
\end{figure}

\begin{figure}[tb]
    \centering
    \scalebox{0.8}{\input{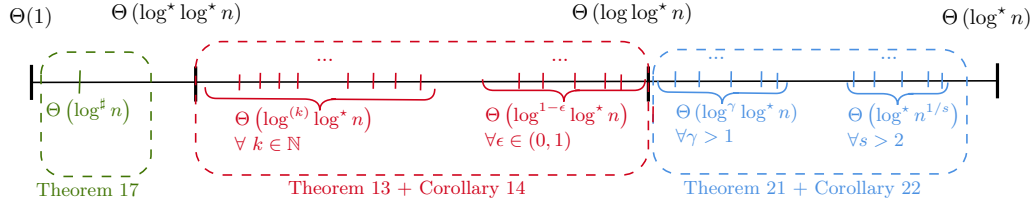}}
    \caption{Complexities obtained by application of our general translation operator.}
    \label{fig:comp_functions2}
\end{figure}

\subsection{Our Techniques}

It may seem at a first glance that  locally checkable problems with complexities in the empty gaps of the complexity landscape of LCL problems can easily be obtained from LCL problems whose complexities are functions of $\Delta$ and~$n$. One typical example is \emph{sinkless orientation}, whose complexity is known to be $\Theta(\log_\Delta n)$~\cite{brandt2016lower,ghaffari2017sinkless}. 
One can define a problem, called \emph{threshold sinkless orientation}, with complexity $t(n)$ by setting a threshold degree $\Delta_0=n^{O(1/t(n))}$. This problem consists of solving sinkless orientation whenever the graph has maximum degree $\Delta\leq \Delta_0$, and doing nothing otherwise. However, this approach suffers from several weaknesses. 
\begin{itemize}
    \item First, for solving the problem, the nodes need to be given the maximum degree $\Delta$ of the input graph, whereas the \LOCAL\ model assumes that the sole global information given to the nodes is a polynomial upper bound on the number of nodes. 

    \item Second, the definition of the problem depends on the number of nodes in the input network, whereas it is desirable that the problem definition (whether it be LCL problems or locally checkable problems) be independent of the size of the input.

    \item Last, but not least, the lower bound for the chosen problem (e.g., sinkless orientation) generally holds only for deterministic algorithms, sometimes to randomized algorithms, but not to quantum algorithms (not to speak about non-signaling  finitely-dependent distribution). 
\end{itemize}
Instead, the approach used in this paper yields problems whose quantum complexities can be lower bounded, which allows us not only to fill up complexity gaps for deterministic algorithms, but also for randomized algorithms. 

The following problem, called \emph{Increasing Degree}, and denoted by $\Pi_f$, plays a major role for obtaining our results. It is parameterized by a function $f:\mathbb{N}\to\mathbb{N}$ satisfying $f(x)>x$ for every $x\in\mathbb{N}$. 
Recall that, in a simple graph $G=(V,E)$, a \emph{leaf}  is a node of degree~1, i.e., a node $v$ such that $\deg(v)=1$, where $\deg(v)=|N(v)|$ with $N(v)=\{w\in V\mid \{v,w\}\in E\}$. 
We denote by $\ldeg(v)$ the number of leaves adjacent to~$v$. 

\begin{definition}
Given $f:\mathbb{N}\to\mathbb{N}$, we say that a node $v$ of $G$ is \emph{balanced} (with respect to~$f$) if it is a leaf, or all its neighbors are adjacent to less than  $f(\ldeg(v))$ leaves, i.e., if
\[
\big(\deg(v)=1\big)\lor \big(\forall w\in N(v), \; \ldeg(w)< f(\ldeg(v))\big). 
\]
A node $v$ which is not balanced is \emph{unbalanced}. For every unbalanced node~$v$, every neighbor ${w\in N(v)}$ satisfying  $\ldeg(w)\geq  f(\ldeg(v))$ is said to be a \emph{heavy} of~$v$ (with respect to~$f$). 
\end{definition}

\begin{figure}[tb]
    \centering
\begin{tikzpicture}[square/.style={regular polygon,regular polygon sides=4}]
   \tikzstyle{circlenode}=[draw,circle,minimum size=70pt,inner sep=0pt]
    \tikzstyle{whitenode}=[draw,circle,fill=white,minimum size=20pt,inner sep=0pt]
    \tikzstyle{whitebox}=[draw,square,fill=white,minimum size=15pt,inner sep=0pt]
    \tikzstyle{blacknode}=[draw=black,circle=black,fill=black,minimum size=3pt,inner sep=0pt]
    \tikzstyle{nonode}=[draw=white,circle=red,fill=white,minimum size=15pt,inner sep=0pt]
 
\draw (-1,0) edge [ dotted, thick]  node {} (-0.5,0);
\draw (0,0) node[whitenode] (a1)  {$v_0$};
\draw (1.5,0) node[whitenode] (a2)   {$v_1$};
\draw (3,0) node[whitenode] (a3)   {$v_2$};
\draw (4.5,0) node[whitenode] (a4)   {$v_3$};
\draw (7,0) node[whitenode] (a6)   {$v_i$};
\draw (8.5,0) node[whitenode] (a7)   {$v_{i+1}$};
\draw (11,0) node[whitenode] (a9)   {$v_k$};

\draw (1.5,-1) node[blacknode] (b1) [] {};
\draw (2.9,-1) node[blacknode] (b2) [] {};
\draw (3.1,-1) node[blacknode] (b3) [] {};
\draw (4.2,-1) node[blacknode] (b4) [] {};
\draw (4.4,-1) node[blacknode] (b5) [] {};
\draw (4.6,-1) node[blacknode] (b6) [] {};
\draw (4.8,-1) node[blacknode] (b7) [] {};
\draw (6.8,-1) node[blacknode] (b9) [] {};
\draw (6.9,-1) node[blacknode] (b10) [] {};
\draw (7,-1) node[blacknode] (b11) [] {};
\draw (7.1,-1) node[blacknode] (b12) [] {};
\draw (7.2,-1) node[blacknode] (b13) [] {};
\draw (8.1,-1) node[blacknode] (b15) [] {};
\draw (8.2,-1) node[blacknode] (b16) [] {};
\draw (8.3,-1) node[blacknode] (b17) [] {};
\draw (8.4,-1) node[blacknode] (b18) [] {};
\draw (8.5,-1) node[blacknode] (b19) [] {};
\draw (8.6,-1) node[blacknode] (b20) [] {};
\draw (8.7,-1) node[blacknode] (b21) [] {};
\draw (8.8,-1) node[blacknode] (b22) [] {};
\draw (8.9,-1) node[blacknode] (b23) [] {};

\draw (10.6,-1) node[blacknode] (b24) [] {};
\draw (10.7,-1) node[blacknode] (b25) [] {};
\draw (10.8,-1) node[blacknode] (b26) [] {};
\draw (10.9,-1) node[blacknode] (b27) [] {};
\draw (11,-1) node[blacknode] (b28) [] {};
\draw (11.1,-1) node[blacknode] (b29) [] {};
\draw (11.2,-1) node[blacknode] (b30) [] {};
\draw (11.3,-1) node[blacknode] (b31) [] {};
\draw (11.4,-1) node[blacknode] (b32) [] {};
\draw (10.6,-1.1) node[blacknode] (b33) [] {};
\draw (10.7,-1.1) node[blacknode] (b34) [] {};
\draw (10.8,-1.1) node[blacknode] (b35) [] {};
\draw (10.9,-1.1) node[blacknode] (b36) [] {};
\draw (11,-1.1) node[blacknode] (b37) [] {};
\draw (11.1,-1.1) node[blacknode] (b38) [] {};
\draw (11.2,-1.1) node[blacknode] (b39) [] {};
\draw (11.3,-1.1) node[blacknode] (b40) [] {};
\draw (11.4,-1.1) node[blacknode] (b41) [] {};

\draw (a1) edge node {} (a2);
\draw (a2) edge node {} (a3);
\draw (a3) edge node {} (a4);
\draw (a4) edge [dotted] node {} (a6);
\draw (a6) edge node {} (a7);
\draw (a7) edge [dotted] node {} (a9);

\draw (a2) edge node {} (b1);

\draw (a3) edge node {} (b2);
\draw (a3) edge node {} (b3);

\draw (a4) edge node {} (b4);
\draw (a4) edge node {} (b5);
\draw (a4) edge node {} (b6);
\draw (a4) edge node {} (b7);

\draw (a6) edge node {} (b9);
\draw (a6) edge node {} (b10);
\draw (a6) edge node {} (b11);
\draw (a6) edge node {} (b12);
\draw (a6) edge node {} (b13);

\draw (a7) edge node {} (b15);
\draw (a7) edge node {} (b16);
\draw (a7) edge node {} (b17);
\draw (a7) edge node {} (b18);
\draw (a7) edge node {} (b19);
\draw (a7) edge node {} (b20);
\draw (a7) edge node {} (b21);
\draw (a7) edge node {} (b22);
\draw (a7) edge node {} (b23);

\draw (a9) edge node {} (b24);
\draw (a9) edge node {} (b25);
\draw (a9) edge node {} (b26);
\draw (a9) edge node {} (b27);
\draw (a9) edge node {} (b28);
\draw (a9) edge node {} (b29);
\draw (a9) edge node {} (b30);
\draw (a9) edge node {} (b31);
\draw (a9) edge node {} (b32);
\draw (a9) edge node {} (b33);
\draw (a9) edge node {} (b34);
\draw (a9) edge node {} (b35);
\draw (a9) edge node {} (b36);
\draw (a9) edge node {} (b37);
\draw (a9) edge node {} (b38);
\draw (a9) edge node {} (b39);
\draw (a9) edge node {} (b40);
\draw (a9) edge node {} (b41);
\end{tikzpicture}
    \caption{A path $v_0,\dots,v_k$ with increasing ``leaf degrees'', i.e., with $\ldeg(v_{i+1})\geq f(\ldeg(v_i))$. The small black nodes are leaves.}
    \label{fig:incr-tree}
\end{figure}
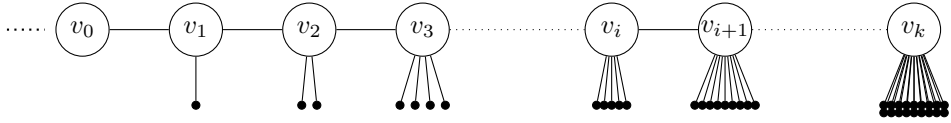

A balanced node $v$ may thus have neighbors $w$ with $\ldeg(w)>\ldeg(v)$, yet the number of leaves adjacent to $w$ cannot be too large, as it must be less than $f(\ldeg(v))$.  In the Increasing Degree problem $\Pi_f$, for a fixed function~$f$, each node $v$ receives a value $\ell_{in}(v)\in\{0,1\}$ as input, and must produce an output $\ell_{out}(v)\in\{0,1\}$. The output of a balanced node must simply be equal to its input, whereas the output of an unbalanced node $v$ must be equal to the output of one of its heavy neighbors~$w$. 

The difficulty for an unbalanced node $v$ to produce its output is that its heavy neighbor $w$ may be unbalanced too, with heavy neighbor~$w'$ which may also be unbalanced, and so on, until a balanced node is found. This chain of heavy neighbors (see Figure~\ref{fig:incr-tree}) may force $v$ to query nodes at long distance from it for producing its output. Formally, the Increasing Degree problem is described as follows:  

\begin{algorithmbox}{\textsc{Increasing Degree Problem $\Pi_f=(L,\varphi_f)$}}
\label{def:increasing_degree_problem}
		
\textbf{Labeling at node $v$:} 
\[
\ell(v)=(\ell_{in}(v),\ell_{out}(v))\in L=\{0,1\}\times \{0,1\} =L_{in}\times L_{out}
\]

\textbf{Correctness predicate $\varphi_f$ at node $v$:} 
\begin{align*}
\varphi_f(v)  = & 
 \Big((\mbox{$v$  balanced}) \land (\ell_{out}(v)=\ell_{in}(v))\Big) \;\; \lor \\
& 
\Big((\mbox{$v$ unbalanced}) \land \big(\exists w\in N(v) :  (\mbox{$w$ heavy w.r.t. $v$}) 
\land  (\ell_{out}(v)=\ell_{out}(w))\big)\Big)
\end{align*}		
\end{algorithmbox}
    
An algorithm solves the Increasing Degree Problem $\Pi_f$ in a graph $G=(V,E)$ if, given any input assignment $\ell_{in}:V\to\{0,1\}$ to the nodes (i.e., node $v$ receives $\ell_{in}(v)$ as input), every node $v$ outputs $\ell_{out}(v)$ such that the predicate $\varphi_f$ is satisfied at all nodes.

\subparagraph*{Remark.} 

Any solution $\ell:V\to\{0,1\}\times\{0,1\}$ to an Increasing Degree Problem $\Pi_f$ is locally checkable at every node $v$. Indeed, assuming every node initially knows its degree, whether $\varphi_f(v)$ is satisfied can be decided based solely on information available in the distance-2 neighborhood of~$v$: At the first round, every node $v$ computes the number of leaves adjacent to it, which is sufficient to check $\varphi_f(v)$ using one additional round. Note that, in trees, one can define a variant of the Increasing Degree problem in which one plays directly with the degrees, and not with the number of leaves. In this case, checking $\varphi_f(v)$ can be done in a single round. 

\medskip

On the other hand, computing the output label at a node may require a large number of rounds, depending on~$f$. Intuitively, the reason for this is that there might exist a long path $v_0,v_1,\dots$ in the input graph~$G$, where, for every $i\geq 1$, $v_i$ is a heavy neighbor of $v_{i-1}$ (see Figure~\ref{fig:incr-tree} for an illustration). As a consequence, the output  $\ell_{out}(v_0)$ may depend on the input label $\ell_{in}(v_k)$ of a node $v_k$ at large distance $k$ from $v_0$. Yet, if $d_i=\ldeg(v_i)$, we have that $d_i\geq f^{(i)}(d_0)$ where $f^{(i)}$ denotes the $i$th iterate of~$f$, and this value cannot exceed~$n$, which in turn upper bounds the value of $k$ as a function of~$n$. Roughly, by picking the smallest value for~$d_0=\ldeg(v_0)$, i.e., $d_0=0$, we show that the complexity of $\Pi_f$ is
\begin{equation}
\label{eq:generic}
    t(n)= \max\Big\{k\geq 0\mid (k+1)+\sum_{i=1}^k f^{(i)}(0)\leq n\Big\}
\end{equation}
whose value is represented by the length of the path $v_0, \dots v_k$ in Figure~\ref{fig:incr-tree}.
The difficulty of the analysis thus comes from the difficulty of solving Equation~\eqref{eq:generic} for a given function~$f$. 
For some specific functions~$f$, this is straightforward, but this is not systematic. Moreover, obtaining problems $\Pi_f$ with very small round complexities requires to choose rapidly growing functions~$f$, for which solving this equation is non trivial. 
This holds even if, for rapidly growing functions~$f$, Equation~\eqref{eq:generic} can be simplified into
\begin{equation}
\label{eq:generic-simple}
    t(n)\simeq \max\Big\{k\geq 0\mid f^{(k)}(0)\leq n\Big\}.
\end{equation}
Again, solving Equation~\eqref{eq:generic-simple} might be straightforward for some specific functions, but not for the vast majority of them. This is where our aforementioned ``translation operator'' $\Phi$ plays an important role for providing a systematic way of generalizing solutions for some functions $f$ to solutions for functions growing much faster than~$f$. Given any \problem $\Pi_f$ with complexity $t_f(n) \leq O(\log n)$, the operator $\Phi$ produces a new version of the \problem, $\Pi_{\Phi(f)}$, with exponentially smaller complexity. 
More specifically, this operator transforms $f$ into the function $\Phi(f)$ where 
\[
\Phi(f)(x) =
\label{tower2}
    {
    \left.
    \begin{array}{c}
        2^{2^{\iddots^{2}}}
    \end{array}
    \right\}
    \text{\footnotesize $f(\log^\star x)$ times}}
\]
from which it follows that $\Pi_{\Phi(f)}$ has complexity 
\[
t_{\Phi(f)}(n) = \Theta\big(t_f(\log^\star n)\big).
\]
Such technique plays an important role: once some version of the \problem with complexities belonging to the interval $[\omega(\log^\star n,), o(\log n)]$ have been found, we are able to find some other correlated problem with complexities in $[\omega(1), o(\log \log^\star n)]$. Note that the two intervals  $[\omega(1), o(\log \log^\star n)]$ and $[\omega(\log^\star n,), o(\log n)]$   are the well-known gaps where no complexity can be found for any general LCL problem (in the bounded degree setting).

\subsection{Related Work}

An excellent entry to the research line studying the  landscape of complexity of LCL problems in the \LOCAL\ model is~\cite{suomela2020landscape}. This landscape has also been studied for specific sub-classes of graphs such as paths, cycles, and toroidal grids~\cite{brandt2017lcl}, paths with inputs~\cite{balliu2019distributed}, trees~\cite{chang2020complexity,grunau2022landscape}, and rooted trees~\cite{balliu2021locally}. It is worth mentioning that another form of complexity landscape, focusing on the size of the certificates required to certify predicates on paths and cycles, has recently been studied in~\cite{bousquet2025complexity}.

To show the existence of problems with complexities belonging to some specific sub-interval of $[0,n-1]$, a common technique is to generate (artificial) problems with specific complexities. A large number of reachable complexities between $\Omega(\log\log^*n)$ and $o(\log n)$, as well as of the form $n^{\Theta(1)}$ and $2^{\log^{\Theta(1)}n}$ have been exhibited in~\cite{balliu2018almost,balliu2018new}. In the case of trees, it was proved~\cite{chang2019time} that there are problems with complexities $n^{1/k}$, for any integer $k\ge2$, but that there is a complexity gap in the interval $[\omega(\log n),o(\mbox{poly}(n))]$. All these constructions assume, for the upper bound on the complexity of each considered problem, that the nodes are given the exact value $n$ of the size of the network. It was proved~\cite{balliu2026distributed} that the existing complexities in trees may vary depending on the degree of knowledge about~$n$, and that randomization helps in absence of any knowledge, which is not the case when $n$ is known. We also refer to~\cite{FraigniaudGKS13} about the impact of the knowledge of $n$ on the ability to perform efficient distributed computation. 
In terms of techniques, a mechanism forcing nodes to transmit information along a path of length~$k$ to get a problem with complexity~$k$ has already been used (see, e.g., \cite{balliu2019distributed,balliu2025shared}).

We complete this brief survey of the related literature by pointing out that the complexity of locally checkable problems when the constraint imposing the maximum degree to be bounded is removed has already been recently studied in~\cite{bousquet2024local,schmid2026lcls}:
\begin{itemize}
    \item In~\cite{schmid2026lcls}, the existence of locally checkable problems with round complexities $\Theta(n^\alpha)$, for any $\alpha\in [0,1]$, has been established. This result is obtained using known facts about algebraic partitions of integers. Our results for the polynomial range reuse some of these ingredients for constructing problems with finer round complexities $\Theta(n^\alpha\log^\beta n)$, in complement to our main results regarding the sub-polynomial regime. 
    
    \item The construction in~\cite{bousquet2024local} enables to obtain locally checkable problems with \emph{arbitrary} round-complexities, whenever one restricts the analysis to trees. The issue of extending their results from trees to arbitrary graphs was explicitly left open in~\cite{bousquet2024local}. Moreover, as underlined in~\cite{schmid2026lcls}, ``the result in~\cite{bousquet2024local} offers little insight into \emph{why} locally checkable problems can behave so erratically whenever the degree constraint is relaxed'', and ``it is unclear which aspects are responsible for the collapse of structure in the complexity landscape for these problems''. Instead, our problems are explicit, simple to describe, and their solutions can be checked locally in a single round. Last but not least, the upper bound in~\cite{bousquet2024local}  works only under the strong condition that the nodes have the knowledge of $n$, for these nodes to compute an upper bound on how many rounds they must run before outputting. Our construction does not require to know $n$ as the number of rounds follows directly from structural conditions.
\end{itemize}
Furthermore, as mentioned before, it is known~\cite{Piombi2026} that the complexity gaps persist even for rooted trees of unbounded degrees, for the wide class of locally checkable problems defined by local Presburger monadic second-order formulas. In the Increasing Degree problems, we use functions such as $x\mapsto x^2$, which cannot be implemented in Persburger arithmetic. Last but not least, in the paper~\cite{damore26superlogarithmicgap} just published on arXiv, it was shown that, in trees, any LCL problem can either be
solved in $O(\log n)$ rounds by a deterministic algorithm, or requires $n^{\Omega(1)}$ rounds to be solved by a
quantum-LOCAL algorithm.

\subsection{Organization of the paper}

Section~\ref{section:warm_up} shows the existence of a locally checkable problem with complexity $\Theta(\log\log n)$, whereas it is known that no LCL problem can have such complexity. It also shows the existence of an \problem with complexity $\Theta(\sqrt{n})$. This latter complexity was known the be achievable by LCL problems, but  Section~\ref{section:warm_up} is also the opportunity to introduce the reader with the \problem at the core of our contribution.  

Section~\ref{sec:sublogn} considers the sublogarithmic gap $[\omega(\log^\star n), o(\log n)]$ that we refer to as the Chang-Pettie Interval~\cite{chang2019time}, whereas Section~\ref{sec:Ramsey-gap} consider the gap $[\omega(1),o(\log\log^*n)]$ that we refer to as the Ramsey interval~\cite{naor1993can}. Section~\ref{sec:gaps-in-trees} is partially dedicated to filling up the known gaps in the complexity of LCL problems in trees, in both polynomial and polylogarithmic regimes, using a simple variant of the Increasing Degree Problem. Section~\ref{sec:refined-gaps} focuses on complexities above $\Omega(\sqrt{n})$, providing refined complexities of the form $\Theta(n^{\frac12+\epsilon}\log^{\alpha}n)$. Finally, Section~\ref{sec:conclusion} concludes the paper with open problems.  

\section{Problems with Respective Complexities $\Theta(\log\log n)$ and $\Theta(\sqrt{n})$}
\label{section:warm_up}

This section can be used as a warm-up, but it already includes an important result, namely the existence of a locally checkable problem with complexity $\Theta(\log\log n)$, whereas it is known that no LCL problem has complexities in the interval $[\omega(\log^\star n),o(\log n)]$. This section also shows the existence of a locally checkable problem with complexity $\Theta(\sqrt{n})$. The existence of such a problem was already known, even for LCL problems, and even if one restricts attention to trees. However, our construction is new, and it provides an introduction to the more sophisticated construction introduced later in the paper. Moreover, our construction does not use any knowledge about the value of~$n$, as opposed to the previous construction~\cite{chang2019time}, about which it has been recently proven in~\cite{balliu2026distributed} that without this knowledge, the complexity changes.

\subsection{A Locally Checkable Problem with Complexity $\Theta(\log\log n)$}
\label{subsection:warm_up-loglog}

\begin{theorem}
\label{thm:logloground}
    There exists a locally checkable problem with round complexity $\Theta(\log\log n)$ in $n$-node graphs. The lower bound $\Omega(\log\log n)$ also holds for trees. 
\end{theorem}

\begin{proof}
    We show that the Increasing Degree Problem $\Pi_f$ with $f(x)=x^2$ satisfies the statement of the theorem. More precisely, let $f:\mathbb{N}\to\mathbb{N}$ defined as 
    \begin{equation}
        \label{eq:x-square}
    f(x)=\left\{\begin{array}{ll}
    1 & \mbox{if $x=0$} \\
    2 & \mbox{if $x=1$} \\
    x^2 & \mbox{otherwise}
    \end{array}\right.
    \end{equation}
    in order to guarantee that $f(x)>x$ for every~$x$.
    To establish the upper bound, we first show the following basic claim. 
    
    \begin{claim}\label{claim:max-length-of-heavy-path}
        For any path $P=(v_0,\dots,v_k)$, if $v_0,\dots,v_{k-1}$ are unbalanced nodes and, for every $i\in\{1,\dots,k\}$, $v_i$ is a heavy neighbor of $v_{i-1}$ with respect to~$f$, then $P$ is of length $k=O(\log\log n)$.
    \end{claim}
    
    To establish the claim,  let $d_i=\ldeg(v_i)$ for every $i\in\{0,\dots,k\}$. 
    For every $i\geq 1$, since $v_i$ is a heavy neighbor of $v_{i-1}$, we have $d_i\geq  d_{i-1}^2$, and thus $d_i\geq d_0^{2^{i+1}}$. (Eq.~\eqref{eq:x-square} guarantees that if $d_0 \in \{0,1\}$, then $d_1 \geq 2$, so the amount of leaf neighbors increases traveling the path even for these cases.) 
    As a consequence, the total number of nodes in $P$ plus the leaves adjacent these nodes is at least 
    \[
    (k+1)+\sum_{i=0}^kd_i = \Omega\left(2^{2^k}\right).
    \]
    Since this number cannot exceed $n$ in $n$-node graphs, we get that the length $k$ of the path cannot exceed $O(\log\log n)$, as claimed.

\subparagraph*{The Algorithm.} 

Our algorithm (see Algorithm~\ref{alg:generic-code2}) does not require to know the aforementioned bound $O(\log\log n)$. Moreover, it does not require node identifiers, nor any information of the number of nodes in the network, and it can be implemented in the port-numbering version of the \LOCAL\ model. It proceeds as follows.  
    After having collected the degrees $\deg(w)$ and the leaf-degree $\ldeg(w)$ of each of its neighbors~$w$, every node $v$ knows whether it is balanced or unbalanced, and, in the latter case, which of its neighbors are heavy. The algorithm then proceeds differently depending on whether $v$ is balanced or not. 
    \begin{itemize}
        \item If $v$ is balanced, then it merely sends its input $\ell_{in}(v)$ to all its neighbors, sets its output $\ell_{out}(v)$ to $\ell_{in}(v)$, and terminate. 

        \item If $v$ is unbalanced, then it waits to receive a value $x\in \{0,1\}$ from one of its heavy neighbors. When it eventually receives such a value~$x$, it forwards it to all its neighbors, and terminates with output $\ell_{out}(v)=x$.
    \end{itemize}

\begin{algorithm}[htb]
\caption{Solving the Increasing Degree Problem $\Pi_f$: Code of node $v$ with input $\ell_{in}(v)\in\{0,1\}$}
\label{alg:generic-code2}
\begin{algorithmic}[1]
\State send $\deg(v)$ to neighbors, and receive $\deg(w)$ from every neighbor $w$
\State send $\ldeg(v)$ to neighbors, and receive $\ldeg(w)$ from every neighbor $w$
\If{$v$ is balanced}
\State send $\ell_{in}(v)$ to all neighbors
\State $\ell_{out}(v)\gets \ell_{in}(v)$
\Comment{\textit{balanced nodes set their outputs and terminate}}
\Else
\Repeat
\State receive messages (if any) from neighbors
\Until{a value $x\in\{0,1\}$ is received from an heavy neighbor}
\State send $x$ to all neighbors
\State $\ell_{out}(v)\gets x$
\Comment{\textit{unbalanced nodes set their outputs and terminate}}
\EndIf \label{inst:phase2-ends}
\end{algorithmic}
\end{algorithm}

\subparagraph*{Correctness.} 

Note first that Algorithm~\ref{alg:generic-code2} guarantees that every balanced node $v$ outputs its input, as desired. To show that every unbalanced node $v$ outputs correctly, let us consider the general scenario of a chain $v_0,v_1,\dots,v_k$ for $k\geq 1$, satisfying that (1)~$v_0=v$, (2)~for every $i\in\{0,\dots,k-1\}$, node $v_i$ is unbalanced, and $v_{i+1}$ is a heavy neighbor of~$v_i$, and (3)~$v_k$ is balanced. Let us show by induction, that node $v_i$ terminates correctly at most $k-i$ rounds after $v_k$ terminated. The result holds for $v_{k-1}$ as it receives the value $x=\ell_{out}(v_k)$ directly from~$v_k$ at the round at which $v_k$ terminates. Assuming the result holds for $0 \leq i<k$, let us show that it holds for $i-1$. Let $t_i\leq i$ be the number of rounds elapsed between $v_k$ and $v_i$ terminate. When $v_{i-1}$ terminates, it is because it has received a value $x\in\{0,1\}$ from one of its heavy neighbors. This occurs at most $t_i+1$ rounds after $v_k$ terminated, as $v_i$ sends its output value $x$ to $v_{i-1}$ before terminating. Therefore, $t_{i-1}\leq t_i+1\leq k-(i-1)$, as claimed. It follows that $v=v_0$ terminates at most $k$ rounds after $v_k$ terminated. 

\subparagraph*{Upper Bound.} 

Since we have proved in Claim~\ref{claim:max-length-of-heavy-path} that no paths $v_0,v_1,\dots$ where $v_i$ is an heavy neighbor of $v_{i-1}$ for all $i\geq 1$ can have a length $k$ exceeding $O(\log\log n)$, we get that Algorithm~\ref{alg:generic-code2} completes in $O(\log\log n)$ rounds, as claimed.


\subparagraph*{Lower Bound.} 

We now show that the Increasing Degree Problem $\Pi_f$ has worst-case round complexity $\Omega(\log\log n)$, even in trees. For this purpose, let us consider the tree $T$ containing a path $P= (v_0, v_1, \dots, v_k)$, with $k=\Omega(\log\log n)$, which we refer to as the \emph{spine} of $T$. Node $v_0$ has one neighbor that is a leaf, and, for every $i\in\{1,\dots,k-1\}$, node $v_i$ is connected to $f^{(i)}(1)=2^{2^{i-1}}$ leaves. By construction, all nodes $v_0,\dots,v_{k-1}$ of the spine are unbalanced, and only $v_k$ is balanced. The lower bound merely follows from an indistinguishability argument: if node $v_0$ performs less than $k$ rounds, and outputs~0 in an instance where $\ell_{in}(v_k)=0$, then it also outputs~0 when  $\ell_{in}(v_k)=1$.  This completes the proof. 
\end{proof}

Observe that, by the same arguments as in the proof of \Cref{thm:logloground}, it can be derived  that the Increasing Degree problem $\Pi_f$ with $f(x)=x^2$ has complexity $\Theta(\log \log \Delta)$ in graphs with maximum degree~$\Delta$, hence constant round complexity whenever $\Delta=O(1)$, as for LCL problems. In fact, this is general for the Increasing Degree Problem $\Pi_f$, for all increasing function~$f$. Indeed, in graphs of maximum degree $\Delta$, the length of a chain of heavy neighbors cannot exceed a constant value $t=t(f,\Delta)$ uniquely defined by $f$ and $\Delta$. Therefore, Algorithm~\ref{alg:generic-code2} applied with $f$ will have complexity at most $t=O(1)$ in graphs of bounded maximum degree. 

\subsection{On Stronger Variants of LOCAL}

The lower bound in the proof of \Cref{thm:logloground} use simple indistinguishability arguments. As such, they also apply to randomized \LOCAL\ algorithms, and to non-signaling algorithms, which include quantum \LOCAL\  algorithms. This holds for all the lower bound proofs in the remainder of the paper. 
The case of the randomized online \LOCAL\ algorithm however needs some more care. Among the several well-studied variants of \LOCAL, the randomized online \LOCAL\ model is indeed the most powerful one. That is, any lower bound that holds in this latter model also holds in any other variants of \LOCAL\ such as randomized-\LOCAL, quantum-\LOCAL\ or \textsf{SLOCAL} \cite{akbari25randonline}.

Recall that, in the online \LOCAL\ model, the nodes of the input graph $G$ are processed sequentially with respect to an order $\sigma = (v_1, v_2, \ldots, v_n)$ arbitrarily chosen by an adversary. 
An algorithm $\AA$ for this model acts as an online algorithm, keeping track of the part of the graph discovered so far.
Specifically, at each step $i\geq 1$, $\AA$~updates its memory to include the portion of the graph $G$ within distance $d\geq 0$ of the node $v_i$ ---~referred to the locality of~$\AA$. $\AA$~has then to produce an irrevocable output label $\ell_{out}(v_i)$ for node $v_i$, which depends of the problem at hand (such a label may be a color, a Boolean indicating whether the node belongs to a maximal independent set, etc.). The algorithm then proceeds to the next step. $\AA$ may be deterministic or randomized.  The proof of Theorem \ref{thm:logloground} can be slightly modified to establish the following. 

\begin{corollary}
\label{cor:stronger-model}
    There exists a locally checkable problem with round complexity $\Theta(\log\log n)$ in $n$-node graphs, where the lower bound $\Omega(\log\log n)$ also holds for trees under the randomized online \LOCAL\ model. 
\end{corollary}

\begin{proof}
    We consider the same problem $\Pi_f$ as in the proof of Theorem \ref{thm:logloground}, and thus the upper bound follows from that theorem. For the lower bound, we show than no randomized online \LOCAL\ algorithms with locality $o(\log\log n)$ can solve~$\Pi_f$. We use the same tree $T$ which contains the path $P=(v_0,..., v_k)$, called the spine of~$T$, with $k = \Omega(\log \log n)$ (see Figure~\ref{fig:incr-tree}). 
    For the purpose of contradiction, let $\mathcal{A}$ be a randomized online \LOCAL\ algorithm solveing $\Pi_f$ with locality $k/3$. 
    Let then $\sigma = (v_k, v_0, v_1,\dots)$ be the sequential order which the adversary selects the nodes. 
    Let us assume that the input label of $v_k$ is $\ell_{in}(v_k)=0$, and that $\ell_{in}(v_{k/2})=1$. 
 
    Due to the structure of~$T$, $v_0$ has to know the output of $v_k$ to produce its own output. 
    After discovering $v_k$ and $v_0$, the adversary can change $T$ into~$T'$, which only differs from $T$ in the existence of an extra node $w$ in the spine $P$ inserted between  $v_{k/2}$ and $v_{k/2 + 1}$. 
    Since the locality of $\AA$ is $k/3$ the algorithm cannot distinguish $T$ from $T'$ after having discovered $v_k$ and $v_0$ only. 
    The presence of $w$ however makes $v_{k/2}$ a balanced node in $T'$. It follows that $v_0$ has to output $\ell_{out}(v_0)=1$ in $T'$, in contrast to the output $\ell_{out}(v_0)=0$ in $T$. Since both $T$ and $T'$ are legal input configuration after two steps of~$\AA$, there is no way for the algorithm to produce a correct output label for $v_0$. 
\end{proof}

\subparagraph*{Remark.}

The proof of Corollary~\ref{cor:stronger-model} is not specific of the complexity $\Theta(\log\log n)$, and it actually holds throughout the paper, for all results using an \problem $\Pi_f$ for establishing some new complexity, independently of the choice of the function~$f$. In particular, it also holds for the  \problem $\Pi_f$  with complexity $\Theta(\sqrt{n})$ described in the next section. 
    
\subsection{A Locally Checkable Problem with Complexity $\Theta(\sqrt{n})$}

\begin{theorem}\label{thm:sqrt_n}
 There exists a locally checkable problem with round complexity $\Theta(\sqrt{n})$ in $n$-node graphs. The lower bound $\Omega(\sqrt{n})$ also holds for trees. 
\end{theorem}

\begin{proof}
We show that the Increasing Degree Problem $\Pi_f$ with $f(x)=x+1$ satisfies the statement of the theorem. We proceed as in the proof of \Cref{thm:logloground}, by upper bounding the length of any chain of heavy neighbors. Let $P=(v_0,\dots,v_k)$ be a simple path such that, for every $i\in\{0,\dots,k-1\}$, $v_i$ is unbalanced, and $v_{i+1}$ is an heavy neighbor of $v_i$. If $d_i=\ldeg(v_i)$, then we have that the path $P$ together with the leaves attached to it amount for at least $\sum_{i=0}^k d_i$ nodes. Since $d_i\geq d_{i-1}+1$ for every $i\geq 1$, and $d_0\geq 0$, we get that even this sum is at least $\sum_{i=1}^ki=\frac12 k(k-1)$, and thus $k=O(\sqrt{n})$. We can then reuse Algorithm~\ref{alg:generic-code} with $f(x)=x+1$ instead of $f(x)=x^2$, and conclude that the complexity of $\Pi_f$ is $O(\sqrt{n})$, as claimed. 
The  lower bound $\Omega(\sqrt{n})$ for trees follows from the same construction as for the lower bound proof in \Cref{thm:logloground}. The constructed spine $v_0,\dots,v_k$ has length $k=\Omega(\sqrt n)$, and the lower bound follows. 
\end{proof}

%
\section{Complexities in the Chang-Pettie Interval}
\label{sec:sublogn}

The approach used to derive the $\Theta(\log\log n)$ and $\Theta(\sqrt{n})$ complexity bounds in \Cref{section:warm_up} consists of defining a suitable function $f\colon \mathbb{N} \to \mathbb{N}$ for which the corresponding Increasing Degree Problem $\Pi_f$ has the desired round complexity. In this section, we show how to apply this technique to get infinitely many complexity values in the Chang-Pettie interval $[\omega(\log^\star n,o(\log n)]$. To simplify the analysis of \Cref{eq:generic}, which will play a major role in the main results of this section, we make the following assumption.

\begin{remark}
\label{remark:fix_val}
  Throughout the paper, whenever considering a function $f\colon \N \to \N$, we assume that $f(0)=1$ and $f(1)=2$ are fixed, and, for $x\geq 2$, the function is defined according to a proper expression, as it was done for the function $f(x)=x^2$ in Eq.~\eqref{eq:x-square}.
\end{remark}

Given a function $f\colon \N \to \N$ with $f(x)>x$ for every~$x$, and a graph $G$, recall that a heavy path in $G$ is a path $P = (v_0,\dots,v_k)$ such that $v_k$ is balanced, every node in $V(P)\smallsetminus v_k$ is unbalanced, and, for every $i \geq 1$, the node $v_i$ is an heavy neighbor of $v_{i-1}$. The round complexity of the Increasing Problem $\Pi_f$ is the length of the longest heavy path. (This directly follows from the design of the generic Algorithm~\ref{alg:generic-code}, which provides the upper bound, and the generic construction in the proof of Theorem~\ref{thm:logloground}, which provides the lower bound.) We thus state the following for further references. 

\begin{observation}
\label{thm:general_complexity}
The round complexity $T_f(n)$ of the Increasing Problem $\Pi_f$ satisfies
\[
T_f(n) = \max\Big \{k\geq 0 \mid (k+1) + \sum_{j=1}^{k} f^{(j)}(0) \leq n\Big\}.
\]
\end{observation}

Furthermore, when $f$ grows rapidly, the expression for $T_f(n)$ in Observation~\ref{thm:general_complexity} is dominated  by the largest term, so we can observe what follows.

\begin{observation}
\label{lem:technical} 
If $ f(x)\geq 2x$ for every $x \in \N$, then $T_f(n)\in \Theta(\max\{k\geq 0\mid f^{(k)}(0)\leq n\})$.
\end{observation}


%

The main result in this section is the following. 

\begin{theorem}
\label{lemma:complexities3}
   There exist locally checkable problems with the following round complexities in $n$-node graphs:
   $\Theta\left(\log^\star n\right)$,
$\Theta\left(\log^{\epsilon}n\right)$ for every $\epsilon\in (0,1)$, and 
$\Theta\left(\log^{(k)}n\right)$ for every integer $k\geq 1$. 
     In all cases, the lower bounds also hold for trees. 
\end{theorem}

The rest of the section is dedicated to the proof of Theorem~\ref{lemma:complexities3}.  
To exhibit \problem with complexity $\Theta(\log^\star n)$, let us consider the function $f\colon\N\to\N$ defined as 
\[
f(x)=2^x
\]
for every $x\geq 0$. Thanks to Observations~\ref{thm:general_complexity} and~\ref{lem:technical} we can focus on determining the maximum $k$ such that $f^{(k)}(0)\leq n$. We have
\[
f^{(k)}(0)
=
\left.
\begin{array}{c}
    2^{2^{\iddots^{2}}}
\end{array}
\right\}
\text{\small $k-2$ times}
\]
Therefore, $f^{(k)}(0)\leq n$ whenever $k\in O(\log^\star n)$, and the largest $k$ satisfying this inequality is in $\Theta(\log^\star n)$. 

\medskip

Let $\epsilon\in(0,1)$. To exhibit a problem with complexity $\Theta\left(\log^{1-\epsilon}n\right)$, let us consider the function $f\colon\N\to\N$ defined as 
\[
f(x)=x\cdot 2^{\log^\epsilon x}
\]
for $x\geq 2$, with the usual convention that $f(0)=1$, and $f(1)=2$. Again, we can focus on determining the maximum $k$ such that $f^{(k)}(0)\leq n$.
        
\begin{claim}
\label{prop:propofF}
For any $k,x\in \N$ with $x\geq 2$, we have 
$f^{(k)}(x) = 2^{y_k(x)},$ 
where 
$y_1(x) = \log x + \log^{\epsilon} x$ and $y_{k+1}(x) = y_k(x) + y_k^\epsilon (x)$ for all $k\geq 1$. 
\end{claim}

\begin{proof} 
By induction on $k$. We have 
\[
f^{(1)}(n)
= n\cdot 2^{\log^\epsilon(n)}
= 2^{\log n}\cdot 2^{\log^\epsilon(n)}
= 2^{\log n + \log^\epsilon(n)}
= 2^{y_1}.
\]
Moreover, assuming that  $f^{(k)}(n) = 2^{y_k(n)}$, we have that
\[
f^{(k+1)}(n)
= f(f^k(n))
= f(2^{y_k(n)}) 
= 2^{y_k(n)}\cdot 2^{\log^\epsilon 2^{y_k(n)}} 
= 2^{y_k(n)+y_k^\epsilon(n)}
=  2^{y_{k+1}(n)}
\]
as claimed.
\end{proof}

\begin{claim}
\label{prop:propofF2}
For every integer $x\ge 2$, the sequence $(y_k(x))_{k\geq 0}$ defined in 
Claim~\ref{prop:propofF}
satisfies $y_k(x)\in \Theta\left(k^{1/(1-\epsilon)}\right)$.
\end{claim}

\begin{proof}
    Let $g(x) = x^{1-\epsilon}$. By the Mean Value Theorem, for every $a<b$ there exists a constant $\beta\in(a,b)$ such that 
    \[
    f(b)-f(a) = f'(\beta)(b-a)
    \]
    where $f'(\beta) = (1-\epsilon)x^{-\epsilon}$.
    Since $x$ is fixed, we denote $y_k(x)$ simply as $y_k$. We get that
    \begin{align*}
                y_{k}^{1-\epsilon} - y_{k-1}^{1-\epsilon}&= g(y_k) - g(y_{k-1}) &\quad\text{since }y_{k-1}\geq y_{k} \\
                & = (1-\epsilon)\beta_k^{-\epsilon}(y_k-y_{k-1}) &\quad \text{for some }\beta_k\in(y_{k-1},y_k)\\
                &= (1-\epsilon)\beta_k^{-\epsilon}y_{k-1}^{\epsilon} &\quad \text{by def. of }y_k\\
                & = (1-\epsilon)\cdot \left(\dfrac{y_{k-1}}{\beta_k}\right)^{\epsilon}
    \end{align*}
    We use the following result to continue. 

\begin{claim}\label{claim:3}
There exists two constants $C_1,C_2\in\mathbb R_{\geq0}$ such that, for every $k\in\N$, and every $\beta_k\in (y_{k-1}, y_k)$, we have
$
C_1\leq (y_{k-1}/\beta_k)^{\epsilon}\leq C_2.
$
\end{claim}

\begin{proof}
Let $k\in \N$ and $\beta_k\in (y_{k-1}, y_k)$. We have
$
y_{k-1}/\beta_k\leq y_{k-1}/y_{k-1} = 1.
$
Therefore, by using $C_1=1$, we prove the lower bound.
Analogously, we have that  
$
y_{k-1}/\beta_k\geq y_{k-1}/y_{k}
$.
Since $y_k = y_{k-1} + y_{k-1}^\epsilon$, we get 
\[
\dfrac{y_{k-1}}{\beta_k}
\geq \dfrac{y_{k-1}}{y_{k-1} + y_{k-1}^\epsilon}
= \dfrac{y_{k-1}}{y_{k-1}(1 + y_{k-1}^{\epsilon-1})}
= \dfrac{1}{1 + y_{k-1}^{\epsilon-1}}.
\]
Since $y_{k-1}\geq y_1$ as $\{y_k\}_{k\in\N}$ is increasing, and since $\epsilon-1\leq 0$, we get that $y_{k-1}^{\epsilon-1}\leq y_1^{\epsilon-1}$, from which we derive that
\[
\dfrac{1}{1 + y_{k-1}^{\epsilon-1}}\geq \dfrac{1}{1 + y_0^{\epsilon-1}}.
\]
Therefore, by taking $C_2 = \left(\dfrac{1}{1 + y_0^{\epsilon-1}}\right)^{-\epsilon}$ we obtain the desired upper bound. 
\end{proof}
            
Going back to the proof of Claim~\ref{prop:propofF2}, we have that 
\[
y_k^{1-\epsilon}-y_{k-1}^{1-\epsilon}  = (1-\epsilon)\cdot \left(\dfrac{y_{k-1}}{\beta_k}\right)^{\epsilon}.
\]
By Claim \ref{claim:3}, we then obtain that 
$
\hat{C_1}\leq y_k^{1-\epsilon}-y_{k-1}^{1-\epsilon}  \leq \hat{C_2}
$
with $\hat{C_1}=(1-\epsilon)$ and $\hat{C_2} = (1-\epsilon)C_2$.
Since the constants $\hat{C_1}$ and $\hat{C_2}$ are independent of $k$, by using $k$ times these last inequalities we get that 
$\hat{C_1}\cdot k + y_{1}^{1-\epsilon}  \leq y_k^{1-\epsilon}\leq \hat{C_2}\cdot k + y_{1}^{1-\epsilon}$.
Finally, since $y_{1}^{1-\epsilon}$ is constant for a fixed~$x$, we eventually get that 
$
{C_1}'\cdot k^{1/(1-\epsilon)} \leq y_k\leq {C_2}'\cdot k^{1/(1-\epsilon)}  
$
for proper constants $C_1'$ and $C_2'$. It follows that $y_k \in\Theta\left(k^{1/(1-\epsilon)}\right)$, completing the proof of Claim~\ref{prop:propofF2}.
\end{proof}

By Claim \ref{prop:propofF}, we get that 
$2^{y_{k_n}(2)}\leq n$, and by Claim~\ref{prop:propofF2} this later inequality is equivalent to
$2^{k^{1/(1-\epsilon)}}\leq n$, 
from which we conclude that if $k$ is the maximum integer satisfying this inequality, then $k\in\Theta\left(\log^{1-\epsilon} n\right)$, concluding the proof for the complexities $\Theta(\log^{1-\epsilon}n)$.

\medbreak

We finally exhibit a problem with complexity $\Theta(\log^{(k)}n)$ for every integer $k\geq 1$. 
The case $k=1$ is known (and can be also obtained by considering $\Pi_f$ with $f(x)=2x$), and the case $k=2$ was established in \Cref{thm:logloground}. Let $k\ge 3$, and let us consider the function $f\colon\N\to\N$ defined as $f(0) =1, f(1)=2$ and, for all $x\ge2$,

\[
f(x) = {
    \left.
    \begin{array}{c}
       2^{2^{2^{\cdot^{\cdot^{2^{\left(\log^{(k-2)} x\right)^2}}}}}}
    \end{array}
    \right\}
    \text{\small $k-2$ twos}}
\]
That is, $f(x)$ is a tower of $k-2$ 2's, with the topmost exponent equal to $\left(\log^{(k-2)} x\right)^2$. Again, we use Observation~\ref{lem:technical}, we look for the largest $k$ such that
$f^{(k)}(2)\leq n$.
By induction in $j\ge1$, we get that, 
for every integer $x\geq 2$, 

\[
f^{(j)}(x) = {
    \left.
    \begin{array}{c}
       2^{
                    2^{
                        2^{
                            \cdot^{
                                \cdot^{
                                    2^{
                                        \left(\log^{
                                            (k-2)
                                        }
                                        x\right)^{
                                            2^j
                                        }
                                    }
                                }
                            }
                        }
                    }
                }
    \end{array}
    \right\}
    \text{\small $k-2$ twos}}
\]
It follows that the largest $k$ satisfies $k\in\Theta(\log^{(k)}n)$. 
This completes the proof of Theorem~\ref{lemma:complexities3}. 
\section{Complexities in the Ramsey Interval}
\label{sec:Ramsey-gap}

Any version of the Increasing Degree problem $\Pi_f$ can be associated with a round complexity $T_f(n)$ given by Observation~\ref{thm:general_complexity}. 
In the previous section, we showed how it is possible to find specific functions $f$ to exhibit many locally checkable problems with different asymptotic complexities in the Chang-Pettie interval
$[\Omega(\logstar n), O(\log n)]$. 
We now exhibit a systematic way to transform such functions $f$ to functions $g$ whose associated Increasing Degree problem $\Pi_{g}$ have different asymptotic complexities in 
$[\Omega(\logstar \logstar n), O(\logstar n)]$, i.e., in the Ramsey interval. 

\begin{theorem}
\label{thm:automatic_transaltion}
    Let $\Pi_f$ be the \problem associated to a function $f$ satisfying $f(x) \geq 2x$ for all $x\in \N$, with complexity $T_f(n)$.
    There exists a function $g$ such that the associated \problem $\Pi_{g}$ has complexity $T_{g}(n) = \Theta(T_f(\logstar n))$.
\end{theorem}

\begin{proof}
     The function $f$ is mapped to the function $g\colon\N\to\N$ satisfying
    \[
    \label{tower2}
    g(x) = {
    \left.
    \begin{array}{c}
        2^{2^{\iddots^{2}}}
    \end{array}
    \right\}
    \text{$f(\log^\star x)$ times.}}
    \]
    Note that $\log^\star  (g(x)) = f(\log^\star x)$.
    Then, composing the function $g$ with itself $k$ times for any $k \in \N$, we have that
    \[
    g^{(k)}(x) = {
    \left.
    \begin{array}{c}
        2^{2^{\iddots^{2}}}
    \end{array}
    \right\}
    \text{$f^{(k)}(\log^\star x)$ times}}.
    \]
    Therefore, $g^{(k)}(0)\leq n$ implies $f^{(k)}(0) \leq \logstar n$, which holds whenever $k\in O(T_f(\log^\star n))$, and the largest $k$ satisfying this inequality is in $\Theta(T_f(\log^\star n))$. 

    For hypothesis, $f(x) \geq 2x$. So $\log^\star(g(x)) = f(\log^\star x) \geq 2\log^\star x \geq \log^\star (2x)$, for any $x \geq 2$. Hence, considering Remark~\ref{remark:fix_val}, $g(x) \geq 2x$ for any $x$.
    Then, from Observation~\ref{lem:technical}, it follows that the round complexity of the \problem $\Pi_g$ is $\Theta(T_f(\log^\star(n))$.
 \end{proof}

The proof of Theorem \ref{thm:automatic_transaltion} defines an automatic translator
$\Phi$ which transforms any function $f$ satisfying $f(x)\geq 2x$ to a new function $\Phi(f)$ satisfying
$$
\Phi(f)(x) = {
    \left.
    \begin{array}{c}
        2^{2^{\iddots^{2}}}
    \end{array}
    \right\}
    \text{$f(\log^\star x)$ times}}
$$
which in turn correspond to a \problem $\Pi_{\Phi(f)}$ with complexity $T_{\Phi(f)}(n) = \Theta(T_f(\logstar n))$.
\cref{sec:sublogn} describes infinitely many versions of the \problem, with complexities in the Chang-Pettie interval. 
\Cref{thm:automatic_transaltion} allows us to derive infinitely many new versions of the \problem, with different asymptotic  complexities  in
$[\Omega(\logstar \logstar n), O(\logstar n)]$.

The operator can even be iterated. Indeed, the function $\Phi(f)$ resulting from the application of the operator grows much faster than $f$, and thus it is possible to apply $\Phi$ on $\Phi(f)$ to obtain a \problem $\Pi_{\Phi\circ\Phi(f)}$ with complexity $T_{\Phi\circ\Phi(f)}(n) \in [\Omega(\logstar\logstar\logstar n), O(\logstar\logstar n)]$. 
In general, by iterating this process as many times as desired, we can obtain infinitely many different asymptotic  complexity in the interval
$[\Omega(\log^{\star^{(k+1)}} n), O(\log^{\star^{(k)}}n)]$,
for every $k \in \N$, where 
$\log^{\star^{(k)}}n$ is the function
$\logstar$ applied $k$ successive times starting from $n$, i.e., $\log^\star(\dots(\log^\star(n))\dots)$ with $k$ applications of $\log^\star$.
Formally, combining \Cref{lemma:complexities3} and \Cref{thm:automatic_transaltion}, yields the following.

\begin{corollary}\label{cor:complexities_sublogstar}
   There exist locally checkable problems with the following round complexities in $n$-node graphs:
$\Theta\left(\logstar \logstar n \right)$,
$\Theta\big(\log^{\epsilon}(\logstar n)\big)$ for every $\epsilon\in (0,1)$, and 
$\Theta\big(\log^{(k)}(\logstar n)\big)$ for every integer $k\geq 1$. 
\end{corollary}

More generally, 

\begin{corollary}
\label{cor:complexities_deeply_sublogstar}
   There exist locally checkable problems with the following round complexities in $n$-node graphs, for every integer $r\geq 1$: 
$\Theta\big(\log^{\star^{(r)}} n \big)$,
$\Theta\big(\log^{\epsilon}(\log^{\star^{(r)}} n)\big)$ for every $\epsilon\in (0,1)$, and 
$\Theta\big(\log^{(k)}(\log^{\star^{(r)}} n)\big)$ for every integer $k\geq 1$. 
\end{corollary}

We conclude this section by exhibiting a locally checkable problem with extremely small complexity, still in $\omega(1)$.

\begin{definition}
The \emph{doubly iterated} logarithm, or log-sharp, is the function that, evaluated on a natural number~$n$, returns the number of times the iterated logarithm function, i.e. the function $\logstar$, must be iteratively applied starting from  $n$ until the result is at most~1. Formally, the function $log^\sharp \colon \N \to \N$ is defined as
\[
\log^\sharp (n) :=
  \begin{cases}
    0                  & \mbox{if } n \le 1; \\
    1 + \log^\sharp(\log^\star n) & \mbox{if } n > 1
   \end{cases}
\]
\end{definition}

\begin{theorem}\label{thm:logsharp}
      There exists a locally checkable problem with round complexity $T_f(n) = \Theta (\log^\sharp n)$.
\end{theorem}

\begin{proof}
Let $f : \mathbb{N} \to \mathbb{N}$ defined as
\[
f(x)= 
\begin{cases}
2 & \text{if } x = 0, \\
2^{f(x-1)} & \text{if } x \geq 1.
\end{cases}
\]
Or, more informally, the function $f$ is defined as
\[
    f(x) = {
    \left.
    \begin{array}{c}      2^{2^{\iddots^{2}}}
    \end{array}
    \right\}
    \text{$x$ times}}
\]
Looking at this latter form, we immediately see that $\logstar (f(x)) = x$. Composing $f$ with itself we obtain:
\[
        f^{(2)}(x) = {
        \left.
        \begin{array}{c}
            2^{2^{\iddots^{2}}}
        \end{array}
        \right\}
        \text{$f(x)$ times}}
\]
and we have $\logstar \logstar (f^{(2)}(x)) = x$. More generally, for any integers $k$ and $x \geq 1$, it holds that
\[
\log^{\star(k)} (f^{(k)}(x)) = x
\]
and
$\log^{\star(k)} (f^{(k)}(0)) = 2$.
As usual now, thanks to Observation~\ref{lem:technical}, we look for the largest integer $k$ satisfying 
$n = \Theta \big( f^{(k)} (0) \big)$.
Applying $k$ times the iterated logarithm to the previous equation we obtain
\[
    \log^{\star(k)} (n) = \log^{\star(k)} (f^{(k)}(0))
    = 2.
\]
Therefore, $k = \Theta(\log^\sharp n) $, as claimed.
\end{proof}

\section{Filling up Polylogarithmic Complexity Gaps}
\label{sec:gaps-in-trees}

In this section, we study round complexities that fill up the known gaps for LCL problems on bounded-degree trees in the polynomial and polylog regimes. The result for the polynomial regime had actually already been established in~\cite{schmid2026lcls}. We provide a new construction, whose interest is that it can be leveraged for deriving new complexities in the $\mbox{polylog}^\star$ regime in graphs (cf. Corollary~\ref{cor:translog2}). 

\begin{theorem}\label{thm:nsandlogs}
The following statement holds.
\begin{enumerate}
        \item For every real $s>2$, there exists a function $f_s\colon\N\to\N$ such that the round complexity of the \problem $\Pi_{f_s}$ is $\Theta(n^{1/s})$.
        
        \item For every real $\gamma>1$, there exists a function $f_\gamma\colon\N\to\N$ such that the round complexity of the \problem $\Pi_{f_\gamma}$ is $\Theta(\log^{\gamma }n)$. 
\end{enumerate}
\end{theorem}

In contrast to the previous sections, it is not direct to specify a function $f$ whose associated problem $\Pi_f$ attains a desired round complexity. 
Instead, we adopt an inverse approach: given a target complexity $T(n) \in [\omega(\log n), o(\sqrt{n})]$, we construct a function $f$ such that $T_f(n) \in \Theta(T(n))$. 
We formalize this approach in the following technical lemma.

\begin{lemma}\label{prop:alphaj}
Suppose that there exists an increasing sequence $(\alpha_j)_{j\ge 1}\subseteq \mathbb{R}_{>0}$ satisfying the following conditions.
\begin{itemize}
    \item $\alpha_j \geq 5$ for every $j\geq 1$,
    \item $\alpha_{j+1}-\alpha_{j}\ge \alpha_{j}- \alpha_{j-1}$ for all $j\ge 2$, and
    \item there exists a function $T\colon\N\to\N$ such that, for every $n\in\N$, the largest integer $k$ satisfying $\sum_{j=0}^{k-1}(k-j)\alpha_j \leq n$ satisfies $k\in \Theta\left(T(n)\right)$
    \end{itemize}
Then there exists a function $f\colon\N\to\N$ such that the round complexity of the \problem $\Pi_{f}$ is $\Theta(T(n))$.
\end{lemma}

\begin{proof}
    Let $(\alpha_j)_{j\ge1}$ be a  sequence satisfying the conditions of the statement. Let us consider functions of the form $f(x) = x+ g(x)$. Our objective is to determe $g$ such that $\Pi_{f}$ has round complexity $\Theta((T(n))$.
    To apply Observation~\ref{thm:general_complexity}, we need the values of $f^{(i)}(0)$ for all $i$. Let $(\beta_i)_{i\ge1}$ defined as
    $\beta_1 = 2+ g_s(2)$, and $\beta_i = \beta_{i-1} + g(\beta_{i-1})$ for all $i\geq2$.
Thanks to a rearrangement of the sum, it holds that 
\begin{equation}\label{eq:equivf}
    (k+1)+\sum_{i=1}^{k}f^{(i)}(0) \leq n\iff (k+4)+\sum_{j=0}^{k-1}((k-2)-j)\beta_j \leq n. 
\end{equation}
We set $\beta_j = \alpha_j$ to derive the function $g$ defined as $g(0)=g(1)=0$, $g(2)=\alpha_1-2$, and,
for every $j\ge2$,
$
g(\alpha_j)=\alpha_j-\alpha_{j-1}.
$
We extend $g$ to all $n\in\N$ by defining $g(n)=g(\alpha_j)$, for every $j\ge1$, and every $n$ satisfying $\alpha_j \le n < \alpha_{j+1}$.
Now consider $f(x) = x + g(x)$ for this function $g$. By the first two conditions satisfied by the sequence $(\alpha_j)_{j\in\N}$ we have that $f(x)> x$ for all $x\in\N$, and $f$ is increasing. Furthermore, the equivalence of Eq.~\eqref{eq:equivf} holds. By the last condition satisfied by the sequence $(\alpha_i)_{i\geq 1}$, we can infer
that, for any integer~$n$, the desired largest integer $k$  is in $\Theta\left(T(n)\right)$. This is because, since $\alpha_i\ge 1$ for all $i$, we have that the expression $(k+4)+\sum_{j=0}^{k-1}((k-2)-j)\alpha_j$ belongs to $\Theta\big(\sum_{j=0}^{k-1}(k-j)\alpha_j\big)$. 
By Observation~\ref{thm:general_complexity},  the round complexity of $\Pi_{f}$ is in $\Theta\left(T(n)\right)$, as claimed.
\end{proof}

We have all the ingredients to prove \Cref{thm:nsandlogs}. 

\begin{proof}[Proof of \Cref{thm:nsandlogs}]
Let us first prove the first item of the statement. 
Let $s\ge2$. By Lemma \ref{prop:alphaj}, it suffices to exhibit an increasing sequence $(\alpha_j)_{j\ge1}$ satisfying the hypotheses of the lemma, with $T(n) = \Theta(n^{1/s})$.
Let $\alpha$ be the sequence defined as $\alpha_j = (4+j)^{s-2}$ for every integer $j\geq 1$. Since $s\ge2$, we have $\alpha_j\geq 5$ for all $j$, and $\alpha_{j+1}-\alpha_j\ge \alpha_j - \alpha_{j-1}$ for all $j\ge2$. Finally, since \(\int_{a}^{b} x^t \,dx=\frac{1}{t+1}(b^{t+1}-a^{t+1})\), we have 
\[
\sum_{j=1}^k (k-j)j^{s-2}\in\Theta\left(k^{s}\right).
\]
We get that, for every $n$, the largest integer satisfying the third condition in the statement of Lemma~\ref{prop:alphaj} is in $\Theta\left(n^{1/s}\right)$. By Lemma \ref{prop:alphaj} we conclude that there exists a function $f_s$ such that the round complexity of the \problem $\Pi_{f_s}$ is $\Theta\left(n^{1/s}\right)$.

\medskip 

We move on with proving the second item of the statement. 
Fix $\gamma>1$ be a real. By Lemma~\ref{prop:alphaj}, it is sufficient to exhibit an increasing sequence
$(\alpha_j)_{j\ge1}$ satisfying with $T(n) = \Theta((\log n)^{\gamma})$,
Let us define $\alpha_j = 2^{j^{1/\gamma}}$ for every $j\ge1$.
Up to exchanging the first two initial values, we can assume w.l.o.g. that
$\alpha_j\ge 5$ for every $j$ and
$
    \alpha_{j+1}-\alpha_j \ge \alpha_j-\alpha_{j-1}
$
for every $j\ge2$. This does not affect the asymptotic estimates below.
For every $k$, let 
\[
    S_k = \sum_{j=1}^{k}(k-j)2^{j^{1/\gamma}}.
\]

\begin{claim}
$S_k \in \Theta
\left(k^{2-21/\gamma}2^{k^{1/\gamma}}\right)$. 
\end{claim}

\begin{proof}
For the upper bound, write $j=k-r$. Since $x\mapsto x^{1/\gamma}$ is concave, it follows from the mean value theorem
that there exists a constant $c>0$ such that, for every $0\le r\le k$,
\[
    k^{1/\gamma}-(k-r)^{1/\gamma} \ge c\,\frac{r}{k^{1-{1/\gamma}}} .
\]
As a consequence, 
\[
    2^{(k-r)^{1/\gamma}}
    \le 2^{k^{1/\gamma}}\exp\!\left(-c'\frac{r}{k^{1-{1/\gamma}}}\right)
\]
for some constant $c'>0$, and thus
\[
S_k
= \sum_{r=0}^{k-1} r\,2^{(k-r)^{1/\gamma}} 
\leq 2^{k^{1/\gamma}}\sum_{r=0}^{k-1}
        r\cdot\exp\left(-c'\frac{r}{k^{1-{1/\gamma}}}\right) 
\leq C\,2^{k^{1/\gamma}}k^{2-2{1/\gamma}}.
\]
For the lower bound, let us restrict the sum to
$1\le r\le c_0 k^{1-{1/\gamma}}$, where $c_0>0$ is a sufficiently small constant.
For such $r$, we have
$
    (k-r)^{1/\gamma} \ge k^{1/\gamma}-C
$
for some constant $C>0$. Thus
$
    2^{(k-r)^{1/\gamma}} \ge c\,2^{k^{1/\gamma}},
$
and consequently
\[
S_k
\ge \sum_{r=1}^{\lfloor c_0 k^{1-{1/\gamma}}\rfloor}
        r\,2^{(k-r)^{1/\gamma}} 
\ge c\,2^{k^{1/\gamma}}
        \sum_{r=1}^{\lfloor c_0 k^{1-{1/\gamma}}\rfloor} r 
\ge c'\,2^{k^{1/\gamma}}k^{2-2{1/\gamma}}.
\]
This completes the proof of the claim. 
\end{proof}

Since the polynomial factor $k^{2-2{1/\gamma}}$ is negligible compared with
$2^{k^{1/\gamma}}$, the largest integer $k$ such that $S_{k}\le n$ satisfies
$
    k^{1/\gamma} = \Theta(\log n),
$
and thus
$
    k = \Theta\bigl((\log^\gamma n)\bigr).
$
By Lemma \ref{prop:alphaj}, there exists a function $f_{\gamma}\colon\N\to\N$ such that the round complexity of $\Pi_{f_{\gamma}}$ is
$
    \Theta\bigl(\log^{\gamma} n\bigr)
$, which completes the proof of Theorem~\ref{thm:nsandlogs}.
\end{proof}

Analogously to \Cref{thm:automatic_transaltion}, one can transfer the round complexities in the range $[\omega(\log n), o(n)]$ (obtained in \Cref{thm:nsandlogs}) to the range $[\omega(\log\log^\star n), o(\log^\star n)]$. However, \Cref{thm:automatic_transaltion} cannot be applied directly in our setting. Indeed, the translation mechanism described there is only valid for functions satisfying $f(x)\ge 2x$, whereas the functions constructed in \Cref{thm:nsandlogs} satisfy $f(x)< 2x$. Therefore, now we describe a new translator that enables us to transfer our desired complexities.

\begin{theorem}\label{thm:translator_two}
    Let $\mathcal F$ be the family of all increasing functions $f\colon\N\to\N$ such that $f(x)\ge x$ and $f(x+1)\ge f(x) +1$. There exists an operator $\Phi_2\colon\mathcal F\to \mathcal F$ with the following property: For every function $f\in\mathcal F$ whose associated \problem $\Pi_f$ has round complexity $T_f(n)\in\omega(\log n)$,  the round complexity of the \problem $\Pi_{\Phi_2(f)}$ is $\Theta\left(T_f(\log^\star n)\right)$.
\end{theorem}

\begin{proof}
    Let $f\in\mathcal F$, define for each $i\in\N$ the partial sums $A_i = f^{(i)}(0)$ and $S_i =\sum_{j=1}^i A_j$.
    Then define the sequence
    
    $$\beta_i = {
    \left.
    \begin{array}{c}
        2^{2^{\iddots^{2}}}
    \end{array}
    \right\}
    \text{$S_i$ times}}$$
    
    and the function $g=\Phi_2(f)$ as $g(\beta_i) = \beta_{i+1}$, $g(x) = \beta_{i+1}$ for all $x\in[\beta_i,\beta_{i+1})$, and  $g(x) = \beta_1$ for all $x\in [0,\beta_1)$.
By definition it holds that $g^{(i)}(0)\in \Theta(\beta_i)$. As $f(x+1)\ge x+1$, the sum $\sum_{i=1}^k g^{(i)}(0)$ is dominated by the last term, and therefore if $k$ is the largest integer for which this sum is at most~$n$, we have $g^{(k)}(0)\in \Theta(n)$. Then  $\log^\star(\beta_k)$ is in $\Theta(\log^\star n)$, and by definition of $\beta_k$
\[
S_{k} = \sum_{i=1}^{k}f^{(i)}(0)\le \log^\star n.
\]
Therefore, by Observation~\ref{thm:general_complexity} we conclude that $k\in \Theta(T_f(\log^\star n))$, from which we conclude that the round complexity of the \problem $\Pi_{g}$ is $\Theta\left(T_f(\log^\star n)\right)$.
\end{proof}

By a direct application of \Cref{thm:nsandlogs} and \Cref{thm:translator_two}, we obtain the following. 

\begin{corollary}\label{cor:translog2}
      The following statement holds.
    \begin{enumerate}
        \item For every real $s>2$, there exists a function $f_s\colon\N\to\N$ such that the round complexity of the \problem $\Pi_{f_s}$ is $\Theta\left((\log^\star n)^{1/s}\right)$.
        \item For every real $\gamma>1$, there exists a function $f_\gamma\colon\N\to\N$ such that the round complexity of the \problem $\Pi_{f_\gamma}$ is $\Theta(\log^{\gamma}\log^\star n)$. 

    \end{enumerate}
\end{corollary}
%
%
\section{Refined Complexities Above $\Omega(\sqrt{n})$}
\label{sec:refined-gaps}

In this section we present a variant of the \problem to obtain new complexities for locally checkable problems, of the form $\Theta( n^\alpha \log^\beta n)$, for $\alpha \in [\frac{1}{2},1)$  and $\beta \geq 1$. This provides a refinement of the complexities of the form $\Theta( n^\alpha)$ already observed in~\cite{schmid2026lcls}. 

The intuition of our construction is the following. To obtain a small complexity $T_f(n)$ for an \problem $\Pi_f$, the associated function $f$ has to grow fast. Conversely, \cref{thm:sqrt_n} show that to obtain a \problem with complexity $\Theta (\sqrt{n})$, we used a slowing growing function, namely $f(x) = x+1$. We therefore face a problem for obtaining a complexity $T_f(n) = \omega(\sqrt{n})$ as the associated function $f:\N\to\N$ has to satisfy $f(x)>x$, hence $f(x) = x+1$ is the minimum that can be achieved under this constraint. 

In order to design a problem with complexity $\omega(\sqrt{n})$, one would like to use a function $f$ satisfying $f(x)= x+o(1)$, but no \problem $\Pi_f$ correspond to such a function. 
To resolve this issue, we borrow an idea from \cite{schmid2026lcls}, which enables to use a particular structure that implicitly translate into using a slowly growing function $f(x)= x+o(1)$. Roughly, we do not count the number of adjacent leaves, but the presence of specific structures in the vicinity. These structures are essentially trees that are associated to  the representation of a partition of an integer.


For the sake of completeness, let us recall some basic elements needed to define the aforementioned structured, following the ideas developed in~\cite{schmid2026lcls}.

\begin{definition}
A \emph{partition} of a natural number $n$ is an ordered sequence of non-necessarily distinct integers 
$\LL = (a_1, a_2, \dots, a_k)$,
where $0<a_1 \leq a_2 \leq \dots \leq a_k$,
satisfying $\sum_{i=1}^k~a_i~=~n$.
We denote by $p(n)$ the number of distinct partitions of $n$. These partitions can be lexicographically ordered, and, for $i\in \{1,...,p(n)\}$, we denote by $\LL_{n,i}$ the $i$-th partition of~$n$. 
\end{definition}

For instance, we have $p(3) = 3$, and $\LL_{3,1} = \{1,1,1\}, \LL_{3,2} = \{1,2\}, \LL_{3,3} = \{3\}$. 


\begin{lemma}[\cite{hardyramanujan1918partitions}]
\label{lem:partitions_estimate}
$
p(n) \sim \frac{1}{4n\sqrt{3}} \exp\!\Bigl(\pi \sqrt{2n/3}\Bigr).
$
\end{lemma}

We now formalize the tree structure from \cite{schmid2026lcls} that we use in our lower bound construction.

\begin{definition}
    Let $\LL_{n,i} = \{a_1, a_2, \dots, a_k\}$ be the $i$th partition of $n\geq 1$. The  \emph{partition tree} $\TT_{n,i}$ associated to this partition is a rooted tree, whose root  $v$ is referred to as the \emph{head} of the tree. The head $v$ has a child $u_j$ for every $j\in\{1,\dots,k\}$, and the node $u_j$ has $a_j - 1\geq 0$ children, which are leaves of $\TT_{n,i}$.
\end{definition}

For instance, the $p(3) = 3$ partitions yields the three trees 
in Figure~\ref{fig:example-partition-tree-n-equal-3}.
Note that every node $v$ can decide in two rounds whether it is the head of a partition tree $\TT_{n,i}$.  

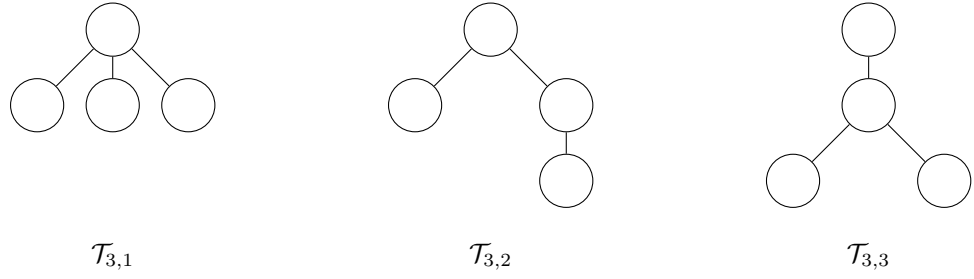
\begin{figure}[tb]
\begin{center}
\begin{tikzpicture}
   \tikzstyle{circlenode}=[draw,circle,minimum size=70pt,inner sep=0pt]
   \tikzstyle{hugenode}=[draw,dotted,circle,minimum size=100pt,fill=black!20, inner sep=0pt]
    \tikzstyle{whitenode}=[draw,circle,fill=white,minimum size=20pt,inner sep=0pt]
    \tikzstyle{greennode}=[draw=red,circle=red,fill=white,minimum size=20pt,inner sep=0pt]
    \tikzstyle{namenode}=[draw=white,circle=white,fill=white,minimum size=15pt,inner sep=0pt]
    \tikzstyle{nonode}=[draw=white,circle=red,fill=white,minimum size=12pt,inner sep=0pt]


    \draw (0,0) node[whitenode] (a1) {};
    \draw (-1,-1) node[whitenode] (b1) {};
    \draw (0,-1) node[whitenode] (c1) {};
    \draw (1,-1) node[whitenode] (d1) {};

\node[] at (0,-3) {$\mathcal{T}_{3,1}$};

\draw (a1) edge node {} (b1);
\draw (a1) edge node {} (c1);
\draw (a1) edge node {} (d1);

    \draw (5,0) node[whitenode] (a2) {};
    \draw (6,-1) node[whitenode] (b2) {};
    \draw (4,-1) node[whitenode] (c2) {};
    \draw (6,-2) node[whitenode] (d2) {};

\node[] at (5,-3) {$\mathcal{T}_{3,2}$};

\draw (a2) edge node {} (b2);
\draw (a2) edge node {} (c2);
\draw (d2) edge node {} (b2);

    \draw (10,0) node[whitenode] (a3) {};
    \draw (10,-1) node[whitenode] (b3) {};
    \draw (9,-2) node[whitenode] (c3) {};
    \draw (11,-2) node[whitenode] (d3) {};

\node[] at (10,-3) {$\mathcal{T}_{3,3}$};

\draw (a3) edge node {} (b3);
\draw (c3) edge node {} (b3);
\draw (b3) edge node {} (d3);

\end{tikzpicture}
\end{center}
\caption{The three trees of depth at most 2 that correspond to the three partitions $(1,1,1),(1,2)$, and $(3)$ of~$n=3$.}
\label{fig:example-partition-tree-n-equal-3}
\end{figure}

We now define an ordering of partition trees, which express how we expect to find a ``chain of partition trees'' in the graph. 

\begin{definition}
\label{def:ordered}
    Let $\alpha \in [\frac{1}{2},1)$ and $\beta \geq 1$ be reals. Let $s := \frac{2\alpha-1}{1-\alpha}$, and $q:= \beta(s+2)$.
    A node $v$ is $(\alpha,\beta)$-\emph{ordered} if the following three conditions hold:
    \begin{enumerate}
        \item $v$ is the head of a partition tree $T_{x,i}$ for some positive integers $i$ and~$x$;
        
        
        \item if $i< \lfloor x^s\log^q x \rfloor$ then $v$ has a neighbor $w$ that is the head of the partition tree $T_{x,i+1}$;
        
        \item if $i = \lfloor x^s\log^q x \rfloor$ then $v$ has a neighbor $w$ that is the head of the partition tree $T_{(x+1),1}$.
    \end{enumerate}
    The node $w$ appearing in the second and third conditions above is called the \emph{successor} of~$v$.
\end{definition}

For every $\alpha\in[\nicefrac12,1)$ and $\beta\geq 1$, we can  now define a new problem $\Pi_{\alpha,\beta}$ that extends the notion of \problem. We call \problemp such a problem $\Pi_{\alpha,\beta}$.

\begin{algorithmbox}{\textsc{\problemp $\Pi_{\alpha,\beta}=(L,\varphi_{\alpha,\beta})$}}
\label{def:increasing_degree_problem}
		
\textbf{Labeling at node $v$:} 
\[
\ell(v)=(\ell_{in}(v),\ell_{out}(v))\in L=\{0,1\}\times \{0,1\} =L_{in}\times L_{out}
\]

\textbf{Correctness predicate $\varphi_{\alpha,\beta}$ at node $v$:} 
\begin{align*}
\varphi_{\alpha,\beta}  = & 
 \Big((\mbox{$v$  not ordered}) \land (\ell_{out}(v)=\ell_{in}(v))\Big) \;\; \lor \\
& 
\Big((\mbox{$v$ ordered}) \land \big(\exists w\in N(v) :  (\mbox{$w$ successor of $v$}) 
\land  (\ell_{out}(v)=\ell_{out}(w))\big)\Big)
\end{align*}		
\end{algorithmbox}

For each \problemp $\Pi_{\alpha,\beta}$, the ordering condition, and the successor nodes are defined with respect to the parameters $\alpha$ and $\beta$ of the problem. 
Similar to an \problem, a \problemp asks every node to merely return its input, except for the nodes which are ordered, which must output the same output as their successors.
Any problem $\Pi_{\alpha,\beta}$ is locally checkable. Indeed, given a possible solution, any node can check in two rounds whether it is the head of a partition tree, and if so whether there it has a successor. The round complexity (in unbounded degree graphs) of a \problemp is the following.

\begin{theorem}\label{thm:complexity_partitions}
    For every 
    $\alpha \in [\frac{1}{2},1)$ 
    and 
    $\beta \geq 1$, the \problemp $\Pi_{\alpha,\beta}$ has round complexity 
    $\Theta(n^\alpha \log^\beta n)$
    in $n$-node graphs. The lower bound 
    $\Omega(n^\alpha \log^\beta n)$
    also holds for trees.
\end{theorem}

\begin{proof}
    Given $\alpha$ and $\beta$ as hypothesis, let
    $s = \frac{2\alpha-1}{1-\alpha}$ and $q= \beta(s+2)$
    as in Definition \ref{def:ordered}. For every positive integer~$k$, let us consider a path $P = (P_1, P_2, \dots, P_k)$ composed of a series of sub-paths $P_1,\dots, P_k$ where 
    \[
    P_i = (v_{i, 1}\ ,\ \dots\ ,\ v_{i, \lfloor i^s\log^q i \rfloor}).
    \]
    In other words, the path $P$ is subdivided in $k$ ``blocks'' $P_1,\dots,P_k$, where the $i$-th block has length  $\lfloor i^s\log^q i \rfloor$. 
    To prove the desired lower bound, we construct a tree $T$ by attaching partition trees to each of the nodes of~$P$. For every $i \in \{1,\dots,k\}$, and every $j\in \{1,\dots,\lfloor i^s\log^q i \rfloor\}$
    we attach the tree $\TT_{i,j}$ to $v_{i,j}$, which becomes the head of $\TT_{i,j}$.
    For every $i\in\mathbb N$, and every $s,q\in\mathbb R$, the tree
    $\TT_{i, j}$
    exists for all $j\leq \lfloor \log^q i \rfloor$, as $i^s\log^{q}(i) \leq p(i)$ by Lemma~\ref{lem:partitions_estimate}. So the tree $T$ is well defined.
    We call $P$ the \emph{spine} of $T$.
    
    Analogously to the proof of \Cref{thm:logloground}, the length $d$ of the spine is a lower bound for the problem $\Pi_{\alpha,\beta}$. We have 
    \[
    d:=|P| = \sum_{i=1}^k i^s\log^q i.
    \]
    Since the function $f(t) = t^s\log^q (t)$ is monotonically increasing in the interval  $[1,\infty)$, for the integral test for convergence we have 
    \[
        \int_1^k f(t)\ dt 
        \leq \sum_{i=1}^k i^s\log^q i
        \leq \int_1^k f(t)\ dt +f(k), 
    \]
    from which it follows that 
    \[
        d 
        = \int_1^k t^s\log^q (t)\ dt + O(k^s\log^q k)
        = \Theta(k^{s+1} \log^q k).
    \]
    On the other hand, the number of nodes of $T$ is
    $n = \sum_{i=1}^k i^{s+1} \log^q i$.
    As a consequence, 
    \[
        n = \int_1^k t^{s+1} \log^q (t)\ dt + O(k^{s+1}\log^q k)
        = \Theta(k^{s+2} \log^q k), 
    \]
    from which it follows that $k = \Theta \big(\frac{n}{\log^q n}\big)^{\frac{1}{s+2}}$. Plugging this latter value in the value of~$d$, we get 
    \[
        d = \Theta \Bigg[\frac
            {n^{\frac
                {s+1}
                {s+2}
                }
            }
            {(\log n)^{\frac
                {q}
                {s+2}
            }}
        \cdot 
        \log^q \Bigg(
            \frac
                {n^{\frac{s+1}{s+2}}}
                {(\log n)^{\frac{q}{s+2}}}
            \Bigg)
        \Bigg]
    = \Theta \bigg(
        C_{s,q}
        \cdot n^{\frac{s+1}{s+2}}
        \cdot \log^{\frac{q}{s+2}} n 
    \bigg)
    \]
    Where $C_{s,q}$ is a value that depends only on $s$ and $q$. Using the values of $\alpha$ and $\beta$, and the fact that both are fixed constants, we eventually get 
    \[
        d = \Theta \Big(
            \frac{\alpha\beta}{1-\alpha}
            \cdot n^\alpha
            \cdot \log^\beta n
            \Big)
        = \Theta \Big(
            n^\alpha
            \log^\beta n
            \Big).
    \]
    The upper bound
    holds because of the construction of the tree $T$ as a $n$-node graph. Indeed it includes $P$, the longest possible path of ordered nodes in a graph with $n$ nodes. Then using an analogue version of \Cref{alg:generic-code} modified for $\Pi_{\alpha,\beta}$ on any graph $G$, the round complexity cannot exceed the one needed for the hard instance $T$.
\end{proof}

\section{Conclusion}
\label{sec:conclusion}

In this paper, we have shown that the gaps in the complexity landscape of LCL problems in the \LOCAL\ model can be filled up with infinitely many different asymptotic complexities of explicit problems whenever the constraint on the maximum degree of the input graphs is discarded. The next step for a better understanding of the complexity landscape of locally checkable problems in the \LOCAL\ model is to prove or disprove the existence of a mechanical way of solving Equation~\eqref{eq:generic-simple} in~$f$ for every given (computable) round complexity~$t(n)$, that is, to explicitly construct $f$ such that the largest integer $k$ satisfying $f^{(k)}(0)\leq n$ is $k=t(n)$. 

As a concrete example, what is an explicit expression of a function $f$ such that the associated \problem $\Pi_f$ has round complexity equal to the inverse Ackermann function? Recall that the inverse Ackermann function is defined as follows. Let $\alpha_1(n)=\lceil n/2 \rceil$, and, for every $k\geq 2$, let $\alpha_k(n)$ be the number of times we have to apply the function $\alpha_{k-1}$ starting from $n$ until we reach a value at most~1, i.e., $\alpha_k(1)=0$, and,  
$
\alpha_k(n)=1+\alpha_k(\alpha_{k-1}(n))
$
for $n\geq 2$.
In particular, $\alpha_2(n)=\lceil\log_2n\rceil$, $\alpha_3(n)=\log^\star n$, and $\alpha_4(n)=\log^\sharp n$. The inverse Ackermann function $\alpha(n)$ assigns to each integer $n$ the smallest $k$ for which $\alpha_k(n)\leq 3$. 
To get a complexity $\alpha_2(n)$, we merely used $\Pi_{f_2}$ with $f_2(x)=2x$, whereas, to get a complexity  $\alpha_3(n)$, we used $\Pi_{f_3}$ with $f_3(x)=2^x$, and, to get a complexity $\alpha_4(n)$, we used $\Pi_{f_4}$ with
\[
    f_4(x) = {
    \left.
    \begin{array}{c}      2^{2^{\iddots^{2}}}
    \end{array}
    \right\}
    \text{$x$ times.}}
\]
More generally, for every index $k\geq 2$, using $\Pi_{f_k}$ with  $f_k(n)=\alpha_{k-1}^{-1}(n)$ enables to define a locally checkable problem with  complexity~$\alpha_k(n)$ in the \LOCAL\ model.  But what is the explicit expression of a function $f_{ack}:\mathbb{N}\to\mathbb{N}$ such that the associated \problem $\Pi_{f_{ack}}$ has complexity~$\alpha(n)$ in the \LOCAL\ model? The same question could be asked for the inverses  of other rapidly growing functions such as the Friedman's TREE function. 

And of course, the general question remains whether \emph{all} round complexities can be obtained in arbitrary graphs, as it was shown~\cite{bousquet2024local} to be the case in the restricted case of trees?

\bibliographystyle{plain}
\bibliography{biblio}
\end{document}